\documentclass[journal,10pt,twoside]{IEEEtran}
\IEEEoverridecommandlockouts

\usepackage{cite}

\usepackage{amsmath,amssymb,amsfonts}
\usepackage{amsthm}
\usepackage{bm}
\usepackage{mathrsfs}

\usepackage{algorithmic}
\usepackage[ruled,linesnumbered]{algorithm2e}

\usepackage{graphicx}
\usepackage{xcolor}
\usepackage{epstopdf}
\graphicspath{{./Figure/}{./Figure/results/}}

\usepackage{textcomp}
\usepackage{enumitem}
\usepackage{url}
\usepackage{verbatim}
\usepackage{multirow}
\usepackage{booktabs}
\usepackage{makecell}
\usepackage{balance}
\usepackage{upgreek}
\usepackage{type1cm}  
\usepackage[T1]{fontenc} 

\usepackage{subcaption}
\usepackage{array}
\usepackage{stfloats}
\usepackage{cuted} 
\usepackage{float}
\usepackage{hyperref}
\hypersetup{
    colorlinks=true,
    linkcolor=blue,
    filecolor=magenta,      
    urlcolor=cyan,
    citecolor=blue,
    }

\newcommand{\herm}{^{\mathsf{H}}}
\newcommand{\trans}{^{\mathsf{T}}}

\DeclareMathOperator{\diag}{\mathsf{diag}}

\allowdisplaybreaks

\makeatletter
\renewcommand{\maketag@@@}[1]{\hbox{\m@th\normalsize\normalfont#1}}%
\makeatother

\newtheorem{theorem}{Theorem}

\title{STAR-RIS-Assisted Integrated Sensing, Secure Communication, and Power Transfer: A Transmit Power Minimization Framework}
\author{Ling He, 
Vaibhav Kumar,  \IEEEmembership{Member, IEEE}, \\
Yingyang Chen, \IEEEmembership{Senior Member, IEEE},
Miaowen Wen, \IEEEmembership{Senior Member, IEEE},\\
Christina P\"opper, \IEEEmembership{Senior Member, IEEE},
and 
Marwa Chafii, \IEEEmembership{Senior Member, IEEE}
\thanks{The work of Vaibhav Kumar, Christina P\"opper, and Marwa Chafii was supported by the Center for Cyber Security through New York University Abu Dhabi Research Institute under Award G1104. An earlier version of this paper was published in part in the Proceedings of the IEEE International Conference on Communications (ICC 2025), Montreal, Canada, in 2025~\cite{25-ICC-SRIS}.}
\thanks{Ling He is with the School of Electronic and Information Engineering, South China University of Technology, Guangzhou 510640, China, and also with the Wireless Research Lab, Engineering Division, New York University Abu Dhabi (NYUAD), UAE (e-mail: eelinghe@mail.scut.edu.cn).}
\thanks{Vaibhav Kumar and Marwa Chafii are with the Wireless Research Lab, Engineering Division, New York University (NYU) Abu Dhabi, UAE. Marwa Chafii is also with NYU WIRELESS, NYU Tandon School of Engineering, New York, USA (e-mail: vaibhav.kumar@ieee.org; marwa.chafii@nyu.edu).}
\thanks{Yingyang Chen is with the College of Information Science and Technology, Jinan University, Guangzhou 510632, China (e-mail: chenyy@jnu.edu.cn).}
\thanks{Miaowen Wen is with the School of Electronic and Information Engineering and the Guangdong Provincial Key Laboratory of Short-Range Wireless Detection and Communication, South China University of Technology, Guangzhou 510640, China (e-mail: eemwwen@scut.edu.cn).}
\thanks{Christina P\"opper is with the Cyber Security \& Privacy (CSP) Lab, Center for Cyber Security, Science Division, New York University Abu Dhabi (NYUAD), UAE (e-mail: christina.poepper@nyu.edu).}
}

\begin{document}
\maketitle

\begin{abstract}
Evolving wireless networks call for architectures that unify sensing, communication, and wireless power transfer. Although integrated sensing and communication (ISAC) and simultaneous wireless information and power transfer (SWIPT) have validated dual-function transmission, the combination of integrated sensing, secure communication, and power transfer (ISSCPT) remains largely unexplored, in part due to the tight coupling among design variables. To address this coupling and expand spatial degrees of freedom, we turn to intelligent metasurfaces: while a conventional reconfigurable intelligent surface (cRIS) reflects only to one side and thus limits coverage and flexibility, a simultaneously transmitting and reflecting RIS (STAR-RIS) enables full-space wave control, making it a natural vehicle for power-efficient ISSCPT. We study a STAR-RIS-assisted ISSCPT system and pose a central question: \emph{How much transmit power is required to operate such a system?} We formulate a transmit-power minimization problem that jointly optimizes transmit and receive beamforming and the STAR-RIS configuration, and solve it via alternating optimization with successive convex approximation, second-order cone programming, and eigenvalue decomposition. Simulations show that the proposed STAR-RIS-assisted design outperforms cRIS and no-RIS baselines, and quantify the additional transmit power required by ISSCPT relative to ISAC and secure SWIPT, clarifying security–sensing–power tradeoffs in metasurface-assisted systems.

\end{abstract}


\newpage 
\section{Introduction}
\IEEEPARstart{T}{he} evolution of wireless networks from fifth-generation (5G) systems to beyond-5G (B5G) and sixth-generation (6G) paradigms is primarily driven by the growing demand for diverse, latency-sensitive, and computationally intensive applications. While 5G focuses on three verticals, namely enhanced mobile broadband (eMBB), massive machine-type communications (mMTC), and ultra-reliable low-latency communications (URLLC), ongoing research on 5G-Advanced (5G-A) emphasizes experience, expansion, extension, and excellence~\cite{23-COMST-6G, 25-ComStdMag-6gBridge}.
Emerging applications include virtual/augmented reality, sub-10~cm positioning, smart grid control, industrial automation, real-time financial transactions, non-terrestrial networks, drone connectivity, broadband access in underserved regions, and the gradual integration of artificial intelligence and machine learning (AI/ML) as enabling technologies. However, enabling these applications demands network architectures that offer ultra-high data rates, precise localization, robust sensing capabilities, and high energy efficiency.

To address these combined demands, early research efforts have explored integrated solutions such as simultaneous wireless information and power transfer (SWIPT)~\cite{18-COMST-SWIPT}, which enables the concurrent delivery of data and energy to different devices. Another promising approach is integrated sensing and communication (ISAC)~\cite{22-CommMag-MultiFunctional6G_ISAC}, which unifies communication and sensing functionalities within a single framework to improve spectral efficiency and resource utilization. Building upon these developments, recent research has further advanced toward integrated sensing, communication, and power transfer (ISCPT)~\cite{24-CommMag-ISCPT}, aiming to enable the joint utilization of spectrum and energy resources within a unified framework. While promising, ISCPT systems encounter significant challenges due to the strong coupling among sensing, communication, and power transfer functions. The resulting interference, resource competition, and design complexity make it challenging to achieve optimal performance across all domains.
% \IEEEPARstart{T}{he} evolution of wireless networks from the fifth-generation (5G) to beyond-5G (B5G) and the sixth-generation (6G) is driven by the growing demand for diverse and increasingly complex applications. While 5G primarily targets three verticals, namely enhanced mobile broadband (eMBB), massive machine-type communications (mMTC), and ultra-reliable low-latency communications (URLLC), ongoing research on 5G-Advanced (5G-A) emphasizes experience, expansion, extension, and excellence~\cite{23-COMST-6G, 25-ComStdMag-6gBridge}. This shift is reflected in emerging services such as virtual/augmented reality, sub-10~cm positioning, smart grid control, industrial automation, real-time financial transactions, non-terrestrial networks, drone connectivity, broadband access in underserved regions, and the gradual integration of artificial intelligence and machine learning (AI/ML) as technology enablers. Meeting the diverse requirements of these applications calls for wireless networks that can simultaneously deliver ultra-high data rates, accurate localization, reliable sensing, and efficient wireless power transfer~\cite{24-CommMag-ISCPT}. To address these challenges, intelligent metasurfaces~\cite{24-JPROC-RIS} have emerged as a promising solution, offering unprecedented flexibility to enhance communication, sensing, and energy harvesting within a unified framework.

Intelligent metasurfaces have emerged as a promising approach to overcome these challenges, offering flexible and programmable control over the wireless environment and enabling additional spatial degrees of freedom for joint optimization~\cite{24-JPROC-RIS}. Conventional reconfigurable intelligent surfaces (cRISs) have been extensively explored for enhancing communication and sensing through passive wavefront reflection~\cite{23-SPM-RIS_ISAC, 23-WC-RIS_ISAC, 25-NetwMag-ISAC_RIS}. Nevertheless, the inherent one-sided nature of cRIS imposes fundamental limitations on coverage and restricts the achievable beamforming performance. To mitigate these limitations, the simultaneously transmitting and reflecting RIS (STAR-RIS) has been recently proposed, enabling $360^{\circ}$ coverage and independent control of transmission and reflection~\cite{24-TVT-Secure_STAR_RIS_SWIPT}. These additional spatial degrees of freedom are particularly advantageous for ISCPT systems. Moreover, security is a key concern in ISCPT systems. Energy receivers (ERs), which are often geographically distributed and require only modest harvested power, may inadvertently or intentionally eavesdrop on confidential communications. Motivated by this challenge, this work proposes an integrated sensing, secure communication, and power transfer framework, hereafter referred to as ISSCPT, and develops a unified beamforming strategy to jointly enhance sensing performance, communication secrecy, and wireless energy transfer.

\subsection{Related Works}
Research on multi-functional wireless systems can be broadly categorized into SWIPT, ISAC, and ISCPT. In SWIPT systems, a single radio-frequency (RF) signal is used simultaneously for information and energy transfer~\cite{18-COMST-SWIPT}. Prior studies have addressed RF energy harvesting (EH) from both ambient and dedicated sources~\cite{18-COMST-SWIPT}, and examined RIS-assisted system design, including passive reflection optimization, channel estimation, and deployment strategies~\cite{22-Proc-RIS_SWIPT}. In particular, RIS-assisted multiple-input multiple-output (MIMO) SWIPT systems have been studied to maximize the weighted sum rate (WSR) at information receivers (IRs) while ensuring sufficient harvested power at ERs~\cite{20-JSAC-SWIPT-MIMO-IRS-Pan, 23-VTC-Kumar-IRS-SWIPT-MIMO, 20-JSAC-RIS_SWIPT_TPM}. Building on these studies, security concerns due to potential eavesdropping by ERs have been extensively investigated in both conventional RIS and STAR-RIS-aided SWIPT systems~\cite{14-TSP-SecureSWIPT_RuiZhang, 23-TWC-SEcure_RIS_SWIPT, 24-TVT-Secure_STAR_RIS_SWIPT, 25-SPL-IRS_Secure_SWIPT}.
Early works~\cite{14-TSP-SecureSWIPT_RuiZhang} laid the foundation by formulating secrecy rate maximization problems for IRs under harvested power constraints at ERs for a downlink multiple-input single-output (MISO) SWIPT system. Extending these studies, a RIS-assisted MISO downlink SWIPT system explored joint transmit beamforming and RIS design~\cite{23-TWC-SEcure_RIS_SWIPT}. Motivated by the $360^{\circ}$ coverage offered by STAR-RIS, subsequent works investigated STAR-RIS-aided SWIPT MISO downlink systems with exclusive information and EH zones, treating all ERs as potential eavesdroppers~\cite{24-TVT-Secure_STAR_RIS_SWIPT}. More recently, secrecy maximization in RIS-assisted underlay spectrum sharing SWIPT systems with multiple eavesdropping ERs has been addressed~\cite{25-SPL-IRS_Secure_SWIPT}. Collectively, these studies highlight the pivotal role of intelligent metasurfaces in enabling secure and efficient joint communication and energy transfer in SWIPT systems.

Meanwhile, the wireless communication community has experienced a paradigm shift toward treating sensing as a core service in next-generation wireless systems, driving growing research interest in ISAC. ISAC system designs are generally categorized as \emph{communication-centric}, \emph{sensing-centric}, or \emph{joint design}. Similar to the case of SWIPT, intelligent metasurfaces have been extensively explored to enhance ISAC performance~\cite{23-SPM-RIS_ISAC, 23-WC-RIS_ISAC, 25-NetwMag-ISAC_RIS, 25-OJCOMS-ISAC-RIS}. Key research areas include waveform design, resource allocation, interference management, channel modeling and estimation, synchronization, cross-layer optimization, and security. Among these, beamforming plays a critical role, as many challenges can be effectively addressed through beamforming optimization. While numerous beamforming strategies have been proposed for communication-centric and sensing-centric ISAC systems, joint beamforming designs remain relatively underexplored. In this direction,~\cite{23-SPL-Pan-JointDesign} proposed a joint active and passive beamforming approach to maximize radar signal-to-interference-plus-noise ratio (SINR) while minimizing communication multi-user interference (MUI). Similarly, a multi-objective optimization to maximize the communication sum rate and minimize the mean-squared error of the sensing beampattern mismatch in a stacked intelligent metasurface (SIM)-aided ISAC was presented in~\cite{25-TVT-ISAC_SIM_MultiObjective}.

Building upon these research directions, some recent studies have explored ISCPT systems. For instance,~\cite{24-CommMag-ISCPT, 25-ComMag-ISCPT} and~\cite{EA-JSAC-O_ISCPT} provided high-level overviews and experimental proofs-of-concept for ISCPT systems. In~\cite{24-JSAC-ISCPT}, the authors addressed the problem of optimal beamforming design to maximize sensing performance while guaranteeing communication and power transfer quality-of-service (QoS) in an ISCPT system with multiple IRs, multiple ERs, and a single target. The work in~\cite{24-TWC-Pareto_ISCPT} analyzed the fundamental performance limits of an ISCPT system consisting of a multi-antenna dual-function base station (DFBS), one multi-antenna IR, one multi-antenna ER, and one target, and characterized the Pareto-optimal trade-off among sensing accuracy, communication rate, and harvested energy. The optimal beamforming design to minimize the required transmit power while satisfying a Cram\'er-Rao bound (CRB) threshold for sensing, as well as achievable rate and energy harvesting constraints at multiple self-powered Internet-of-Things (IoT) devices was presented in~\cite{24-COMML-ISCPT_TPM}. For OFDM-based ISCPT systems,~\cite{25-TWC-ISCPT_OFDM} considered a beamforming design to minimize the average beampattern matching error for sensing, while meeting average communication rates at IRs and average harvested power at ERs. Similarly,~\cite{25-TCOM-ISCPT} investigated an ISCPT system with a hybrid analog-digital architecture, jointly optimizing beamforming and dynamic on-off control to minimize BS power consumption under communication rate, CRB, and energy harvesting constraints. 

At the same time, the potential of intelligent metasurfaces for enhancing ISCPT systems has been widely recognized and actively explored~\cite{24-WCL-WCL-ISCPT_RIS, 25-TWC-SecureISCPT_RIS, EA-IoTJ-ISCPT_IRS, EA-TWC-ISCPT_RIS, EA-TCCN-ISCPT-Active_STAR_RIS, EA-TCOM-ISCPT_aRIS}. For instance,~\cite{24-WCL-WCL-ISCPT_RIS} considered a cRIS-aided ISCPT system, optimizing the harvested energy at ERs while satisfying CRB and signal-to-noise ratio (SNR) constraints. An active-RIS-aided \emph{secure} ISCPT system was studied in~\cite{25-TWC-SecureISCPT_RIS} to maximize the total harvested power, subject to secrecy rate and integrated sidelobe level ratio (ISLR) requirements.
A cRIS-aided ISCPT system consisting of multiple single-antenna self-powered communication users and multiple targets was considered in~\cite{EA-IoTJ-ISCPT_IRS}, where the optimal beamforming design was proposed to maximize the minimum harvested energy at the communication users while satisfying the SINR constraint at the communication users and the  beampattern gain at the targets. The authors in~\cite{EA-TWC-ISCPT_RIS} investigated an RIS-assisted wireless-powered sensing and communication system to maximize the weighted sum of communication rate and the beampattern gain. Building on this, an active-RIS and STAR-RIS-aided ISCPT system was proposed in~\cite{EA-TCCN-ISCPT-Active_STAR_RIS} to maximize the sum rate of the IRs, while guaranteeing a predefined QoS for ERs and the targets, using minimum mean square error (MMSE) transformation, alternating optimization (AO), and successive convex approximation (SCA). 
Similarly, the authors in~\cite{EA-TCOM-ISCPT_aRIS} investigated an active-RIS-aided ISCPT system consisting of multiple targets and multiple self-powered IRs, and maximized the achievable sum rate while ensuring the sensing SNR requirements for the targets and the harvested-energy constraints at the IRs.
% Similarly, the authors in~\cite{EA-TCOM-ISCPT_aRIS} considered an active-RIS-aided ISCPT system, consisting of multiple targets and multiple self-powered IRs, and maximized the achievable sum rate, while satisfying sensing SNR for the targets and harvested energy at the IRs, using the quadratically-constrained quadratic program (QCQP), AO, majorization-minimization (MM), and SCA. 
Despite these advances, jointly optimizing beamforming for sensing, secure communication, and power transfer remains an open challenge, motivating the present study on ISSCPT.

\subsection{Contributions and Organization}
It is important to note that although the works presented in~\cite{24-WCL-WCL-ISCPT_RIS, 25-TWC-SecureISCPT_RIS, EA-IoTJ-ISCPT_IRS, EA-TWC-ISCPT_RIS, EA-TCCN-ISCPT-Active_STAR_RIS, EA-TCOM-ISCPT_aRIS} laid a solid foundation for beamforming design in metasurface-assisted ISCPT systems, the works in~\cite{24-WCL-WCL-ISCPT_RIS, 25-TWC-SecureISCPT_RIS, EA-IoTJ-ISCPT_IRS} followed a power-transfer-centric design,~\cite{EA-TWC-ISCPT_RIS} focused on a communication-and-sensing-centric design, and~\cite{EA-TCCN-ISCPT-Active_STAR_RIS, EA-TCOM-ISCPT_aRIS} were based on a communication-centric design; \emph{the beamforming strategy for a joint design} in a metasurface-aided ISSCPT system is still unexplored. A joint design strategy is of particular interest for practical deployment, as it enables simultaneous optimization of sensing, communication, security, and power transfer functionalities, whereas sensing-centric, communication-centric, or power-transfer-centric designs typically prioritize one function at the expense of the others. It is also noteworthy that the security of a metasurface-assisted ISCPT system was only investigated in~\cite{25-TWC-SecureISCPT_RIS}, where the targets were treated as potential eavesdroppers and the IRs and ERs were assumed to be co-located due to the consideration of self-powered IRs. 
However, when IRs and ERs are geographically separated and the ERs typically require only small amounts of power, even minor information leakage to the ERs can pose a significant security threat to the IRs’ messages. 
%Although prior studies have investigated metasurface-assisted ISAC and SWIPT systems, the secure beamforming design for metasurface-enabled ISCPT systems, where energy receivers may act as eavesdroppers, remains largely unexplored. 
In this work, we propose a novel STAR-RIS–assisted ISSCPT architecture and develop a unified beamforming optimization strategy that jointly enhances communication security, sensing performance, and energy transfer efficiency.
% However, to the best of our knowledge, none of the existing works have considered the problem of beamforming design in a \emph{metasurface-assisted} secure ISCPT system where ERs are potential eavesdroppers. To fill these research gaps, in this paper, we consider a STAR-RIS-aided secure ISCPT system with multiple IRs, multiple eavesdropping ERs, and multiple targets, and propose a joint design strategy for beamforming optimization. 

The main contributions of this paper are summarized as follows: 
\begin{itemize}
    \item We consider a STAR-RIS-enabled ISSCPT system with imperfect channel state information (CSI), consisting of a multi-antenna DFBS, a STAR-RIS, multiple single-antenna IRs, multiple single-antenna eavesdropping ERs, and multiple targets. To exploit a joint design strategy, we formulate the problem of minimizing the transmit power from the DFBS, while satisfying SINR constraints for the IRs, harvested power constraints for the ERs, information leakage constraints at the ERs, echo SINR constraints for target sensing, and power-splitting constraints for each of the STAR-RIS meta-atoms.
    
    \item The formulated optimization problem is challenging to solve due to the non-convexity of all the constraints and coupling between the design variables (i.e., transmit beamforming vectors, receive beamforming vectors, and STAR-RIS beamforming). To obtain a numerically efficient solution, we adopt an AO-based approach, where for a fixed transmit and STAR-RIS beamforming, we update the receive beamforming vectors in closed-form, and then for the fixed receive beamforming vectors, we simultaneously update the transmit and STAR-RIS beamforming vectors via SCA and second-order cone programming (SOCP). 
     
    \item  
    Simulation results demonstrate that the proposed STAR-RIS–aided scheme achieves a superior balance among communication security, sensing accuracy, and energy transfer efficiency, outperforming both cRIS-assisted and no-RIS systems. Furthermore, the results quantify the integration gains and performance trade-offs relative to STAR-RIS–aided ISAC and secure SWIPT baselines.
    % Extensive simulation results are presented to evaluate the performance of the proposed solution and to highlight the influence of various design parameters on the system performance. With the help of the presented results, we show the advantage of the STAR-RIS-aided design over its cRIS-assisted and no-RIS-aided counterparts, and also quantify the cost of integration over STAR-RIS-aided ISAC and STAR-RIS-aided secure SWIPT systems. 
\end{itemize}

The organization of this paper is as follows. The system model and problem formulation for the STAR-RIS-enabled ISSCPT system are introduced in Section~II. In Section~III, we provide a solution to the formulated optimization problem for minimizing transmit power. Simulation results are presented in Section~IV. Finally, Section~V  concludes the paper.  

\textit{Notation}: 
Uppercase boldface and lowercase boldface letters represent matrices and vectors, respectively. $ (\cdot)\trans $ and $ (\cdot)\herm $ represent (ordinary) transpose and Hermitian transpose, respectively. $ \| \cdot\|_2 $ and $ \| \cdot  \|_\mathsf{F}$ respectively denote the $ \ell_2 $-norm and Frobenius norm. $ \mathbf{I} $ represents the identity matrix, and $\boldsymbol{0}$ denotes an all-zero matrix. $ x \sim \mathcal{CN}(a, b) $ means that the scalar $x$ follows a complex Gaussian distribution with mean $a$ and variance $b$. $\mathbb E\{\cdot\}$ is the expectation operator. $ \diag (\mathbf{a})  $ denotes a diagonal matrix with the elements of vector  $\mathbf{a}$ on its main diagonal. For a complex-valued matrix $\mathbf X$, $\Re\{\mathbf X\}$ and $\Im\{\mathbf X\}$ represent the real and imaginary components of $\mathbf X$, respectively. $\mathcal{O}(\cdot)$ represents the Bachmann–Landau notation. $\mathbf X^{(\upsilon)}$ denotes the value of the optimization variable $\mathbf X$ in the $\upsilon$-th iteration. Finally, $ \mathbf{X} \in \mathbb{C}^{M\times N} $ indicates that $ \mathbf{X} $ is a complex-valued matrix with dimensions $ {M\times N} $.

\section{System Model and Problem Formulation} 

We consider the STAR-RIS-enabled ISSCPT system shown in~Fig.~\ref{Fig_system}, consisting of a DFBS (B), a STAR-RIS (S), $I$ IRs, $E$ ERs, and $L$ targets. Assume that B is equipped with $M_{\mathrm{T}}$ transmit antennas and $M_{\mathrm{R}}$ receive antennas, S is equipped with $M_{\mathrm S}$ meta-atoms, all the IRs and ERs are single-antenna nodes, while the targets are passive nodes. 
By $\mathbf{H}\in\mathbb{C}^{M_{\mathrm S}\times M_{\mathrm{T}}}$, $\mathbf{h}_{\mathrm{B},i}\in\mathbb{C}^{1\times M_{\mathrm{T}}}$, and $\mathbf{g}_{\mathrm{B},e}\in\mathbb{C}^{1\times M_{\mathrm{T}}}$,
we denote the wireless links from B to S, B to the $i$-th IR, and B to the $e$-th ER, respectively; here $i\in\mathcal{I}\triangleq\{1,2,\ldots,I\}$ and $e\in\mathcal{E}\triangleq\{1,2,\ldots,E\}$. 
Moreover, we use $\mathbf{h}_{\mathrm S,i}\in\mathbb{C}^{1\times M_{\mathrm S}}$, and $\mathbf{g}_{\mathrm S,e}\in\mathbb{C}^{1\times M_{\mathrm S}}$ to denote the wireless channels from S to the $i$-th IR, and S to the $e$-th ER, respectively. 
Similar to~\cite{24-WCL-WCL-ISCPT_RIS, EA-IoTJ-ISCPT_IRS, 24-TWC-RIS-BF, 24-TWC-SRIS-Act-BF}, we assume
that the links between S and targets do not exist due to blockage or large distances.
The echo channel for the BS  $\rightarrow$ the $l$-th target  $\rightarrow$ the BS is given by $\mathbf{V}_l \in\mathbb{C}^{M_{\mathrm{R}}\times M_{\mathrm{T}}}$ with $l\in\mathcal{L}\triangleq\{1,2,\ldots,L\}$.
% The transmit steering vector from B to $l$-th target is denoted by $\mathbf{v}_{\mathrm{T},l}\in\mathbb{C}^{1\times M_{\mathrm{T}}}$ with $l\in\mathcal{L}\triangleq\{1,2,\ldots,L\}$,
% and the receive echo steering vector from the $l$-th target to the
% receive antenna array at B is denoted by $\mathbf{v}_{\mathrm{R},l}\in\mathbb{C}^{M_{\mathrm{R}}\times1}$.
The residual self-interference link between the transmit and receive antenna arrays at B due to the full-duplex operation in monostatic sensing is denoted by $\mathbf{G}\in\mathbb{C}^{M_{\mathrm{R}}\times M_{\mathrm{T}}}$.
% Without loss of generality (WLOG), in this paper, we assume that the BS lies in the reflection region of the STAR-RIS{\color{blue}~\cite{25-STARS-MIMO}}. 
Without loss of generality (WLOG), in this paper, we assume that the BS is positioned on the front side of the STAR-RIS, corresponding to its reflection region~\cite{25-STARS-MIMO}.
For the system under consideration, B transmits a linear combination of information signals (intended for IRs), energy harvesting symbols (intended for ERs) and probing signals for sensing the targets, which is given by  
% the information receivers are interested in the communication message transmitted from B, while the energy receivers are interested in harvesting the wireless energy/power contained in the wireless transmission from B. The signal transmitted from B is given by 
\begin{align}
\mathbf{x}=\underbrace{\sum\limits_{i\in\mathcal{I}}\mathbf{f}_{\mathrm{I},i}w_{\mathrm{I},i}}_{\text{Communication signal}}+\underbrace{\sum\limits_{e\in\mathcal{E}}\mathbf{f}_{\mathrm{E},e}w_{\mathrm{E},e}}_{\text{Energy harvesting signal}}+\underbrace{\sum\limits_{l\in\mathcal{L}}\mathbf{f}_{\mathrm{T},l}w_{\mathrm{T},l}}_{\text{Probing signal}},\label{eq:tx-signal}
\end{align}
where $w_{\mathrm{I},i}$ is the information-bearing complex-valued constellation symbol intended for the $i$-th IR, $w_{\mathrm{E},e}$ is the energy harvesting signal for the $e$-th ER, and $w_{\mathrm{T},l}$ is the probing signal for the $l$-th target. 
The beamforming vector for the $i$-th IR, the $e$-th ER, and the $l$-th target are respectively denoted by $\mathbf{f}_{\mathrm{I},i}\in\mathbb{C}^{M_{\mathrm{T}}\times1}$, $\mathbf{f}_{\mathrm{E},e}\in\mathbb{C}^{M_{\mathrm{T}}\times1}$, and $\mathbf{f}_{\mathrm{T},l}\in\mathbb{C}^{M_{\mathrm{T}}\times1}$, respectively.

\begin{figure}[t] 
	\centering
	\includegraphics[width=0.9\columnwidth]{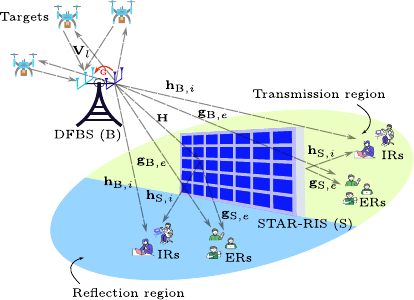}
	\caption{System model for the STAR-RIS-enabled ISSCPT.}
	\label{Fig_system}
\end{figure}

In this paper, we assume that B has imperfect CSI estimates about
$\mathbf{h}_{\mathrm{B},i}$, $\mathbf{h}_{\mathrm S,i}$, $\mathbf{g}_{\mathrm{B},e}$,
$\mathbf{g}_{\mathrm S,e}$, and $\mathbf{V}_{l}$,
respectively denoted by $\widehat{\mathbf{h}}_{\mathrm{B},i}$, $\widehat{\mathbf{h}}_{\mathrm S,i}$,
$\widehat{\mathbf{g}}_{\mathrm{B},e}$, $\widehat{\mathbf{g}}_{\mathrm S,e}$, and $\widehat{\mathbf{V}}_{l}$,
with the corresponding CSI errors denoted by $\bm{\delta}_{\mathrm{B},i}\sim\mathcal{CN}(\bm{0},\varsigma_{\mathrm{B},i}^{2}\mathbf{I})$,
$\bm{\delta}_{\mathrm S,i}\sim\mathcal{CN}(\bm{0},\varsigma_{\mathrm S,i}^{2}\mathbf{I})$,
$\bm{\delta}_{\mathrm{B},e}\sim\mathcal{CN}(\bm{0},\varsigma_{\mathrm{B},e}^{2}\mathbf{I})$,
$\bm{\delta}_{\mathrm S,e}\sim\mathcal{CN}(\bm{0},\varsigma_{\mathrm S,e}^{2}\mathbf{I})$,
and $\bm{\Delta}_{l}\sim\mathcal{CN}(\bm{0},\varsigma_{l}^{2}\mathbf{I})$.
WLOG, in the rest of this paper, we assume that $\varsigma_{\mathrm{B},i}^{2}=\varsigma_{\mathrm S,i}^{2}=\varsigma_{\mathrm{B},e}^{2}=\varsigma_{\mathrm S,e}^{2}=\varsigma_{l}^{2}=\varsigma^{2},\forall i\in\mathcal{I},e\in\mathcal{E},l\in\mathcal{L}$. 

\subsection{Communication Model at the IRs}
The received signal at the $i$-th IR is given by 
\begin{align}
y_{\mathrm{I},i}=\  & \big(\mathbf{h}_{\mathrm{B},i}+\mathbf{h}_{\mathrm S,i}\bm{\Theta}_{\mathrm{I},i}\mathbf{H}\big)\mathbf{x}+n_{\mathrm{I},i}\nonumber \\
% =\  & \big[\widehat{\mathbf{h}}_{\mathrm{B},i}+\bm{\delta}_{\mathrm{B},i}+\big(\widehat{\mathbf{h}}_{\mathrm S,i}+\bm{\delta}_{\mathrm S,i}\big)\bm{\Theta}_{\mathrm{I},i}\mathbf{H}\big]\mathbf{x}+n_{\mathrm{I},i}\nonumber \\
=\  & \big(\underbrace{\widehat{\mathbf{h}}_{\mathrm{B},i}+\widehat{\mathbf{h}}_{\mathrm S,i}\bm{\Theta}_{\mathrm{I},i}\mathbf{H}}_{\triangleq\widehat{\mathbf{z}}_{\mathrm{I},i}}\big)\mathbf{x}+\big(\underbrace{\bm{\delta}_{\mathrm{B},i}+\bm{\delta}_{\mathrm S,i}\bm{\Theta}_{\mathrm{I},i}\mathbf{H}}_{\triangleq\bm{\delta}_{\mathrm{I},i}}\big)\mathbf{x}+n_{\mathrm{I},i}\nonumber \\
=\  & \widehat{\mathbf{z}}_{\mathrm{I},i}\mathbf{x}+\bm{\delta}_{\mathrm{I},i}\mathbf{x}+n_{\mathrm{I},i},\label{eq:rx-signal_Ii}
\end{align}
where $n_{\mathrm{I},i}\sim\mathcal{CN}(0,\sigma_{\mathrm{I},i}^{2})$
is the additive white Gaussian noise (AWGN) at the $i$-th IR. Moreover, $\bm{\Theta}_{\mathrm{I},i}=\diag(\bm{\theta}_{\mathrm{I},i})$
with $\bm{\theta}_{\mathrm{I},i}$ being the STAR-RIS response
for the $i$-th IR. In this paper, we consider the power-splitting
protocol at the STAR-RIS, \emph{i.e.}, for an incident signal $x$, the response
of the $\varkappa$-th STAR-RIS meta-atom $\big(\varkappa\in\mathcal{M}_{\mathrm S}\triangleq\{1,2,\ldots,M_{\mathrm S}\}\big)$ in the reflection and transmission
regions are respectively given by $\theta_{\mathsf{R},\varkappa}$
and $\theta_{\mathsf{T},\varkappa}$, provided $\angle\theta_{\mathsf{R},\varkappa},\angle\theta_{\mathsf{T},\varkappa}\in[0,2\pi)$
and $\sum_{\mathsf{m}\in\{\mathsf{R},\mathsf{T}\}}|\theta_{\mathsf{m},\varkappa}|^{2}=1$.
Therefore, we can define $\bm{\theta}_{\mathrm{I},i}$ as 
\begin{align}
\bm{\theta}_{\mathrm{I},i}=\mathbb{A}_{\mathrm{I},i}\bm{\theta}_{\mathsf{R}}+(1-\mathbb{A}_{\mathrm{I},i})\bm{\theta}_{\mathsf{T}},\label{eq:theta_Ii_def}
\end{align}
where the indicator function $\mathbb{A}_{\mathrm{I},i}=1$ if the
$i$-th IR lies in the reflection region of the STAR-RIS, and $\mathbb{A}_{\mathrm{I},i}=0$
otherwise. Here $\bm{\theta}_{\mathsf{R}}=[\theta_{\mathsf{R},1},\ldots,\theta_{\mathsf{R},M_{\mathrm S}}]\trans\in\mathbb{C}^{M_{\mathrm S}\times1}$,
and $\bm{\theta}_{\mathsf{T}}=[\theta_{\mathsf{T},1},\ldots,\theta_{\mathsf{T},M_{\mathrm S}}]\trans\in\mathbb{C}^{M_{\mathrm S}\times1}$.
We denote $\mathbf{F}=[\mathbf{f}_{\mathrm{I},1},\ldots,\mathbf{f}_{\mathrm{I},I},\mathbf{f}_{\mathrm{E},1},\ldots,\mathbf{f}_{\mathrm{E},E},\mathbf{f}_{\mathrm{T},1},\ldots,\mathbf{f}_{\mathrm{T},L}]\in\mathbb{C}^{M_{\mathrm{T}}\times(I+E+L)}$
as the stack of all transmit beamforming vectors and $\bm{\theta}=\big[\bm{\theta}_{\mathsf{R}}\trans,\bm{\theta}_{\mathsf{T}}\trans\big]\trans\in\mathbb{C}^{2M_{\mathrm S}\times1}$
as the STAR-RIS beamforming vector. With a slight abuse of notations,
we redefine the stack of transmit beamforming vectors as $\mathbf{F}\triangleq\big[\mathbf{f}_{1},\ldots,\mathbf{f}_{I},\mathbf{f}_{I+1},\ldots\mathbf{f}_{I+E},\mathbf{f}_{I+E+1},\ldots\mathbf{f}_{I+E+L}\big]$.
We assume that the $i$-th IR is interested only in decoding $w_{\mathrm I,i}$, and hence, the SINR for decoding the intended signal $w_{\mathrm{I},i}$ at the
$i$-th IR can be given by
\begin{align}
&\gamma_{\mathrm{I},i}(\mathbf{F},\bm{\theta})=   \frac{\big|\widehat{\mathbf{z}}_{\mathrm{I},i}\mathbf{f}_{i}\big|^{2}}{\underbrace{\sum\limits_{\jmath\in\mathcal{Q}\setminus \{i\}}\big|\widehat{\mathbf{z}}_{\mathrm{I},i}\mathbf{f}_{\jmath}\big|^{2}}_{\text{MUI}}+\underbrace{\sum\limits_{q\in\mathcal{Q}}\mathbb{E}\big\{\big|\bm{\delta}_{\mathrm{I},i}\mathbf{f}_{q}\big|^{2}\big\}}_{\text{ICE}}+\sigma_{\mathrm{I},i}^{2}}\nonumber \\
&\!= \!\frac{\big|\widehat{\mathbf{z}}_{\mathrm{I},i}\mathbf{f}_{i}\big|^{2}}{\sum\limits_{\jmath\in\mathcal{Q} \setminus \{i\}} \! \big|\widehat{\mathbf{z}}_{\mathrm{I},i}\mathbf{f}_{\jmath}\big|^{2} \!+\!\varsigma^{2}\!\sum\limits_{q\in\mathcal{Q}}\!\Big(\big\Vert\mathbf{f}_{q}\big\Vert^{2}\!+\!\big\Vert\bm{\Theta}_{\mathrm{I},i}\mathbf{H}\mathbf{f}_{q}\big\Vert^{2}\Big)\!+\!\sigma_{\mathrm{I},i}^{2}},
% \\
% =\  & \big|\widehat{\mathbf{z}}_{\mathrm{I},i}\mathbf{f}_{i}\big|^{2}\Big[\sum\limits_{\jmath\in\mathcal{Q}/i}\big|\widehat{\mathbf{z}}_{\mathrm{I},i}\mathbf{f}_{\jmath}\big|^{2}\nonumber \\
%  & +\varsigma^{2}\sum\limits_{q\in\mathcal{Q}}\Big(\big\Vert\mathbf{f}_{q}\big\Vert^{2}+\big\Vert\bm{\Theta}_{\mathrm{I},i}\mathbf{H}\mathbf{f}_{q}\big\Vert^{2}\Big)+\sigma_{\mathrm{I},i}^{2}\Big]^{-1},
 \label{eq:gamma_Ii}
\end{align}
where $\mathcal{Q}\triangleq\{1,2,\ldots,Q\}$, $Q = I+E+L$ and ICE represents the interference due to CSI errors. 

\subsection{Power Harvesting Model at the ERs}
At the same time, the signal received at the $e$-th ER is given by
\begin{align}
&y_{\mathrm{E},e}=\big(\mathbf{g}_{\mathrm{B},e}+\mathbf{g}_{\mathrm S,e}\bm{\Theta}_{\mathrm{E},e}\mathbf{H}\big)\mathbf{x}+n_{\mathrm{E},e}\nonumber \\
% =\  & \big[\widehat{\mathbf{g}}_{\mathrm{B},e}+\bm{\delta}_{\mathrm{B},e}+\big(\widehat{\mathbf{g}}_{\mathrm S,e}+\bm{\delta}_{\mathrm S,e}\big)\bm{\Theta}_{\mathrm{E},e}\mathbf{H}\big]\mathbf{x}+n_{\mathrm{E},e}\nonumber \\
&= \big(\underbrace{\widehat{\mathbf{g}}_{\mathrm{B},e}+\widehat{\mathbf{g}}_{\mathrm S,e}\bm{\Theta}_{\mathrm{E},e}\mathbf{H}}_{\triangleq\widehat{\mathbf{z}}_{\mathrm{E},e}}\big)\mathbf{x}+\big(\underbrace{\bm{\delta}_{\mathrm{B},e}+\bm{\delta}_{\mathrm S,e}\bm{\Theta}_{\mathrm{E},e}\mathbf{H}}_{\triangleq\bm{\delta}_{\mathrm{E},e}}\big)\mathbf{x}+n_{\mathrm{E},e}\nonumber \\
&= \widehat{\mathbf{z}}_{\mathrm{E},e}\mathbf{x}+\bm{\delta}_{\mathrm{E},e}\mathbf{x}+n_{\mathrm{E},e},\label{eq:rx_signal_Ee}
\end{align}
where $n_{\mathrm{E},e}\sim\mathcal{CN}(0,\sigma_{\mathrm{E},e}^{2})$
is the AWGN at the $e$-th ER. Moreover, $\bm{\Theta}_{\mathrm{E},e}=\diag(\bm{\theta}_{\mathrm{E},e})$,
and 
\begin{align}
\bm{\theta}_{\mathrm{E},e}=\mathbb{A}_{\mathrm{E},e}\bm{\theta}_{\mathsf{R}}+(1-\mathbb{A}_{\mathrm{E},e})\bm{\theta}_{\mathsf{T}},\label{eq:theta_Ee_def}
\end{align}
where the indicator function $\mathbb{A}_{\mathrm{E},e}=1$ if the
$e$-th ER lies in the reflection region of the STAR-RIS, and $\mathbb{A}_{\mathrm{E},e}=0$
otherwise. Hence the total harvested power at the $e$-th ER is given
by
\begin{align}
 & P_{\mathrm{E},e}(\mathbf{F},\bm{\theta})=\eta_{\mathrm{E},e}\sum\limits_{q\in\mathcal{Q}}\Big(|\widehat{\mathbf{z}}_{\mathrm{E},e}\mathbf{f}_{q}|^{2}+\mathbb{E}\big\{|\bm{\delta}_{\mathrm{E},e}\mathbf{f}_{q}|^{2}\big\}\Big)\nonumber \\
&=\   \eta_{\mathrm{E},e}\sum\limits_{q\in\mathcal{Q}}\Big(|\widehat{\mathbf{z}}_{\mathrm{E},e}\mathbf{f}_{q}|^{2}+\varsigma^{2} \Big\{\big\Vert\mathbf{f}_{q}\big\Vert^{2}+\big\Vert\bm{\Theta}_{\mathrm{E},e}\mathbf{H}\mathbf{f}_{q}\big\Vert^{2}\Big\}\Big),\label{eq:harvested_power_Ee}
\end{align}
where $0<\eta_{\mathrm{E},e}<1$ is the power harvesting efficiency
of the $e$-th ER, and WLOG, it is assumed that the noise power at the ERs is negligibly small with respect to (w.r.t.) the power harvested at the ERs. 

\subsection{Information Leakage Model at the ERs}
In this paper, we assume that ERs are potential eavesdroppers, and each of the ERs attempts to decode the information symbols, $w_{\mathrm I, i}$. The SINR at the $e$-th ER to decode the signal intended for the $i$-th
IR, \emph{i.e.}, $w_{\mathrm{I},i}$ is given by 
\begin{align}
&\gamma_{\mathrm{E},e,i}(\mathbf{F},\bm{\theta})=  \frac{\big|\widehat{\mathbf{z}}_{\mathrm{E},e}\mathbf{f}_{i}\big|^{2}}{\underbrace{\sum\limits_{\jmath\in\mathcal{Q}\setminus \{i\}}|\widehat{\mathbf{z}}_{\mathrm{E},e}\mathbf{f}_{\jmath}|^{2}}_{\text{MUI}}+\underbrace{\sum\limits_{q\in\mathcal{Q}}\mathbb{E}\big\{\big|\bm{\delta}_{\mathrm{E},e}\mathbf{f}_{q}\big|^{2}\big\}}_{\text{ICE}}+\sigma_{\mathrm{E},e}^{2}}\nonumber \\
&\!=\!\frac{\big|\widehat{\mathbf{z}}_{\mathrm{E},e}\mathbf{f}_{i}\big|^{2}}{\sum\limits_{\jmath\in\mathcal{Q} \setminus \{i\}}\!\!\big|\widehat{\mathbf{z}}_{\mathrm{E},e}\mathbf{f}_{\jmath}\big|^{2}\!+\!\varsigma^{2}\!\!\sum\limits_{q\in\mathcal{Q}}\!\!\Big(\big\Vert\mathbf{f}_{q}\big\Vert^{2}\!+\!\big\Vert\bm{\Theta}_{\mathrm{E},e} \mathbf{H}\mathbf{f}_{q}\big\Vert^{2}\Big)\!+\!\sigma_{\mathrm{E},e}^{2}}.
% =\  & \big|\widehat{\mathbf{z}}_{\mathrm{E},e}\mathbf{f}_{i}\big|^{2}\Big[\sum\limits_{\jmath\in\mathcal{Q}/i}\big|\widehat{\mathbf{z}}_{\mathrm{E},e}\mathbf{f}_{\jmath}\big|^{2}\nonumber \\
%  & \!+\!\varsigma^{2}\!\sum_{q\in\mathcal{Q}}\!\Big(\big\Vert\mathbf{f}_{q}\big\Vert^{2}\!+\!\big\Vert\bm{\Theta}_{\mathrm{E},e}\mathbf{H}\mathbf{f}_{q}\big\Vert^{2}\Big)\!+\!\sigma_{\mathrm{E},e}^{2}\Big]^{-1}.
 \label{eq:gamma_Ee}
\end{align}

\subsection{Sensing Model at the DFBS}
We now model the sensing performance metric in terms of the received echo SINR at B. For this purpose, the received echo signal from the $l$-th target at B can be given by 
\begin{align}
\mathbf{y}_{\mathrm{B}}=\  & \sum\nolimits_{l\in\mathcal{L}}\alpha_{l}\mathbf{V}_{l}\mathbf{x}+\mathbf{G}\mathbf{x}+\mathbf{n}_{\mathrm{B}}\nonumber \\
=\  & \sum\nolimits_{l\in\mathcal{L}}\alpha_{l}\left(\widehat{\mathbf{V}}_{l} + \bm{\Delta}_{l} \right) \mathbf{x}+\mathbf{G}\mathbf{x}+\mathbf{n}_{\mathrm{B}},\label{eq:yB}
\end{align}
where $\alpha_{l}$ is the radar cross-section (RCS) of the $l$-th target with $\mathbb{E}\{|\alpha_{l}|^{2}\}=\bar{\alpha},\forall l\in\mathcal{L}$ and $\mathbf{n}_{\mathrm{B}}\sim\mathcal{CN}(\mathbf{0},\sigma_{\mathrm{B}}^{2}\mathbf{I})$ is the AWGN vector at the receive antenna array at B. 
% and we have neglected the term $\alpha_{l}\bm{\delta}_{\mathrm{R},l}\bm{\delta}_{\mathrm{T},l}\mathbf{x}$ since its value is much smaller than other terms.
Before processing the received signal, B applies a receive beamforming vector $\mathbf{c}_{l}\in\mathbb{C}^{M_{\mathrm{R}}\times1}$ to sense the $l$-th target, resulting in the following post-combining echo SINR for the $l$-th target:
\begin{align}
 & \gamma_{\mathrm{B},l}(\mathbf{F},\mathbf{c}_{l})=\bar{\alpha}\sum\limits_{q\in\mathcal{Q}}\big|\mathbf{c}_{l}^{\mathsf{H}}\widehat{\mathbf{V}}_{l}\mathbf{f}_{q}\big|^{2}\Big[\bar{\alpha}\sum\limits_{\jmath\in\mathcal{L} \setminus \{l\}}\sum\limits_{q\in\mathcal{Q}}\big|\mathbf{c}_{l}^{\mathsf{H}}\widehat{\mathbf{V}}_{\jmath}\mathbf{f}_{q}\big|^{2}\nonumber \\
&+\sum\limits_{q\in\mathcal{Q}}\Big(\bar{\alpha}\varsigma^{2}\big\Vert\mathbf{c}_{l}\big\Vert^{2}\big\Vert\mathbf{f}_{q}\big\Vert^{2}+\big|\mathbf{c}_{l}^{\mathsf{H}}\mathbf{G}\mathbf{f}_{q}\big|^{2}\Big)+\sigma_{\mathrm{B}}^{2}\big\Vert\mathbf{c}_{l}\big\Vert^{2}\Big]^{-1}.\label{eq:gamma_Bl}
\end{align}
We further define the stack of receive beamforming vectors as $\mathbf{C}\triangleq[\mathbf{c_{1}},\ldots,\mathbf{c}_{L}]\in\mathbb{C}^{M_{\mathrm{R}}\times L}$. 

\subsection{Problem Formulation}
With the given background, we now formulate the following optimization problem:\footnote{We assume that the targets are passive nodes without transceiver capability, and hence do not pose any security threat to the network. Likewise, the IRs are considered trusted users. These assumptions simplify the formulated problem, while the more general case with eavesdropping targets and untrusted users is left for future research.} 
\begin{subequations}
\label{eq:P1}
\begin{align}
\underset{\mathbf{F},\bm{\theta},\mathbf{C}}{\min}\  & \|\mathbf{F}\|_{\mathsf{F}}^{2} \label{eq:P1-objective}\\
\text{s.t.}\  & \gamma_{\mathrm{I},i}(\mathbf{F},\bm{\theta})\geq\Gamma_{\mathrm{I},i},\forall i\in\mathcal{I},\label{eq:P1-comm}\\
 & P_{\mathrm{E},e}(\mathbf{F},\bm{\theta})\geq\mathcal{P}_{e},\forall e\in\mathcal{E},\label{eq:P1-EH}\\
 & \gamma_{\mathrm{E},e,i}(\mathbf{F},\bm{\theta})\leq\Gamma_{\mathrm{E},e,i},\forall e\in\mathcal{E},i\in\mathcal{I},\label{eq:P1-leakage}\\
 & \gamma_{\mathrm{B},l}(\mathbf{F},\mathbf{c}_{l})\geq\Gamma_{\mathrm{B},l},\forall l\in\mathcal{L},\label{eq:P1-sensing}\\
 & \|\mathbf{c}_{l}\|^{2}=1,\forall l\in\mathcal{L},\label{eq:P1-unit-norm-beamforming}\\
 & \sum\nolimits_{\mathsf{m}\in\{\mathsf{R},\mathsf{T}\}}|\theta_{\mathsf{m},\varkappa}|^{2}=1,\forall\varkappa\in\mathcal{M}_{\mathrm S}.\label{eq:P1-UMC}
\end{align}
\end{subequations}
In~\eqref{eq:P1}, we aim to minimize the total transmit power from B, while maintaining a given communication QoS constraint for every IR, a power harvesting QoS constraint for every ER, an information leakage
QoS constraint at every ER corresponding to every information symbol, a sensing QoS constraint corresponding to every target, and the power-splitting protocol at the STAR-RIS. More specifically,~\eqref{eq:P1-objective} is the total transmit power from B, \eqref{eq:P1-comm} ensures that the SINR at the $i$-th IR is greater than or equal to the predefined threshold $\Gamma_{\mathrm{I},i}$, \eqref{eq:P1-EH} enforces that the total harvested power at the $e$-th ER is greater than or equal to $\mathcal{P}_{e}$, \eqref{eq:P1-leakage} ensures that the SINR of decoding $w_{\mathrm{I},i}$ (\emph{i.e.}, symbol intended for the $i$-th IR) at the $e$-th ER (\emph{i.e.}, the information leakage) is less than or equal to $\Gamma_{\mathrm{E},e,i}$, \eqref{eq:P1-sensing} ensures that the sensing SINR corresponding to the $l$-th target at B is at least $\Gamma_{\mathrm{B},l}$, \eqref{eq:P1-unit-norm-beamforming} enforces that echo signals are not amplified at the receive antenna array of B, and~\eqref{eq:P1-UMC} is due to the power-splitting protocol at the STAR-RIS.\footnote{Although a more practical coupled transmission–reflection phase-shift model for STAR-RIS was proposed in~\cite{22-JSTSP-STAR}, this work adopts the widely used decoupled phase-shift model to ensure analytical tractability.} 

In the following section, we propose an AO-based numerically efficient solution to~\eqref{eq:P1}, and present the corresponding convergence and complexity analyses of the proposed algorithm.
% \footnote{\textcolor{red}{Mention that coupled phase-shift model is not considered
% here.}} 

% Moreover, $\mathcal{M}_{\mathrm S}\triangleq\{1,2,\ldots,M_{\mathrm S}\}$. 

\section{AO-Based Proposed Solution}
% The problem in~\eqref{eq:P1} is non-convex due to the coupling between
% the optimization variables $\mathbf{F}$ and $\bm{\theta}$
% in~\eqref{eq:P1-objective}–\eqref{eq:P1-leakage}, that between
% $\mathbf{F}$ and $\mathbf{c}_{l}$ in~\eqref{eq:P1-sensing}, and
% the equality constraints in~\eqref{eq:P1-unit-norm-beamforming}
% and \eqref{eq:P1-UMC}. 
The optimization problem in~\eqref{eq:P1} is a nonconvex quadratically constrained quadratic program (QCQP). Its nonconvexity arises entirely from the constraints: the QoS constraints in~\eqref{eq:P1-comm}–\eqref{eq:P1-leakage} involve SINR and harvested-power expressions that are quadratic functions jointly w.r.t. the transmit beamforming matrix $\mathbf{F}$ and the STAR-RIS coefficients $\bm{\theta}$, while the sensing QoS constraints in~\eqref{eq:P1-sensing} contain bilinear couplings between $\mathbf{F}$ and the receive combiners $\mathbf{c}_{l}$. In addition, the unit-norm and power-splitting equalities in~\eqref{eq:P1-unit-norm-beamforming} and~\eqref{eq:P1-UMC} are also nonconvex quadratic constraints, so that the feasible set is nonconvex and standard convex optimization techniques cannot be applied directly.
To solve this challenging non-convex optimization
problem, we propose an AO-based iterative approach, where the optimal
receive beamforming vectors ($\mathbf{C}$) are obtained in closed-form
via eigenvalue decomposition, and then we simultaneously optimize
$\mathbf{F}$ and $\bm{\theta}$ using a sequence of convex
approximations.\footnote{It is worth emphasizing that the SCA framework developed in this work could, in principle, be extended to jointly update the entire variable tuple $(\mathbf{F},\bm{\theta},\mathbf{C})$ within a single convex surrogate problem at each iteration. However, we intentionally adopt a block coordinate / AO strategy, in which $(\mathbf{F},\bm{\theta})$ are updated via SCA whereas $\mathbf{C}$ is computed in closed-form, to exploit the problem structure, maintain manageable per-iteration computational complexity, and avoid solving an unnecessarily high-dimensional convex program.}
% we provide a detailed solution to the problem in~\eqref{eq:P1}, using an AO-based approach. $\left(\mathscr{P}_{1}\right)$ is decomposed into 

\subsection{Optimal Receive Beamforming for Fixed \texorpdfstring{$(\mathbf{F},\bm{\theta})$}{ (F, theta)}}

Observe that the receive beamforming vectors enter the formulation only through~\eqref{eq:P1-sensing} and~\eqref{eq:P1-unit-norm-beamforming}. Using~\eqref{eq:gamma_Bl}, the post-combining sensing SINR can be written as
\begin{align}
\gamma_{\mathrm{B},l}(\mathbf{c}_{l})=\frac{\mathbf{c}_{l}^{\mathsf{H}}\mathbf{M}_{1,l}\mathbf{c}_{l}}{\mathbf{c}_{l}^{\mathsf{H}}\mathbf{M}_{2,l}\mathbf{c}_{l}},\label{eq:gamma_Bl_transform}
\end{align}
where 
\begin{align}
\mathbf{M}_{1,l}\triangleq\bar{\alpha}\sum\nolimits_{q\in\mathcal{Q}}\widehat{\mathbf{V}}_{l}\mathbf{f}_{q}\mathbf{f}_{q}^{\mathsf{H}}\widehat{\mathbf{V}}_{l}^{\mathsf{H}},\label{eq:M1l}
\end{align}
and 
\begin{multline}
\mathbf{M}_{2,l}\triangleq  \bar{\alpha}\sum\nolimits_{\jmath\in\mathcal{L} \setminus \{l\}}\sum\nolimits_{q\in\mathcal{Q}}\widehat{\mathbf{V}}_{\jmath}\mathbf{f}_{q}\mathbf{f}_{q}^{\mathsf{H}}\widehat{\mathbf{V}}_{\jmath}^{\mathsf{H}} \\
 +\bar{\alpha}\varsigma^{2}\big\Vert\mathbf{F}\big\Vert_{\mathsf{F}}^2\mathbf{I} +\sum\nolimits_{q\in\mathcal{Q}}\mathbf{G}\mathbf{f}_{q}\mathbf{f}_{q}^{\mathsf{H}}\mathbf{G}^{\mathsf{H}}+\sigma_{\mathrm{B}}^{2}\mathbf{I}.\label{eq:M2l}
\end{multline}
In~\eqref{eq:gamma_Bl_transform}, we suppress the explicit dependence of the sensing SINR on $\mathbf{F}$ on the left-hand side (LHS), since the transmit beamforming vectors are treated as fixed parameters while optimizing over $\mathbf{c}_{l}$. For fixed $(\mathbf{F},\bm{\theta})$, maximizing $\gamma_{\mathrm{B},l}(\mathbf{c}_{l})$ subject to the unit-norm constraint in~\eqref{eq:P1-unit-norm-beamforming} constitutes a fractional quadratic program over the complex unit sphere. Since $\mathbf{M}_{1,l}$ and $\mathbf{M}_{2,l}$ are Hermitian and $\mathbf{M}_{2,l}\succ\mathbf{0}$ (due to the noise and interference terms), this subproblem can be equivalently expressed as a \emph{generalized Rayleigh quotient maximization} and admits a \emph{closed-form global solution} given by the principal generalized eigenvector of the matrix pair $(\mathbf{M}_{1,l},\mathbf{M}_{2,l})$. Equivalently, the optimal receive combiner at the $\upsilon$-th iteration can be given by
\begin{align}
\mathbf{c}_{l}^{(\upsilon)}=\bm{\lambda}_{\max}\big(\mathbf{M}_{2,l}^{-1}\mathbf{M}_{1,l}\big),\label{eq:cl_opt}
\end{align}
where $\bm{\lambda}_{\max}(\mathbf{X})$ denotes the generalized principal eigenvector of the matrix $\mathbf{X}$. By convention, the eigenvectors are taken to have unit Euclidean norm, so~\eqref{eq:cl_opt} satisfies the constraint in~\eqref{eq:P1-unit-norm-beamforming} by construction.

\subsection{Optimal \texorpdfstring{$(\mathbf{F},\bm{\theta})$}{ (F, theta)} for Fixed \texorpdfstring{$\mathbf{C}$}{ C}}

In this subsection, we use a sequence of convex approximation to simultaneously update $\mathbf{F}$ and $\bm{\theta}$, while $\mathbf{C}$ is fixed. Note that for a given $\mathbf{C}$,~\eqref{eq:P1} can be reformulated as follows: 
\begin{align}
\underset{\mathbf{F},\bm{\theta}}{\min}\ \big\{\|\mathbf{F}\|_{\mathsf{F}}^{2}\mid\eqref{eq:P1-comm}-\eqref{eq:P1-sensing},\eqref{eq:P1-UMC}\big\}.
\label{eq:P2}
\end{align}
% \eqref{eq:P1-comm},\eqref{eq:P1-EH},\eqref{eq:P1-leakage},\eqref{eq:P1-sensing}
It is easy to note that the objective in~\eqref{eq:P2} is already convex, and due to the minimization problem, no transformation is required. We thus shift our focus to convexify the constraints in~\eqref{eq:P1-comm}. 
\begin{theorem}
\label{thm:comm-Convex}For a fixed $i\in\mathcal{I}$, a convex equivalent
of~\eqref{eq:P1-comm} is given by 
\begin{subequations}
\label{eq:P1-comm-CONVEX}
\begin{align}
 & \frac{1}{\Gamma_{\mathrm{I},i}}\psi_{\mathrm{I},i}\big(\mathbf{f}_{i},\bm{\theta};\mathbf{f}_{i}^{(\upsilon)},\bm{\theta}^{(\upsilon)}\big)\geq\sum\nolimits_{\jmath\in\mathcal{Q} \setminus \{i\}}\big(\tau_{\mathrm{I},\jmath}^{2}+\bar{\tau}_{\mathrm{I},\jmath}^{2}\big) +\varsigma^{2}\nonumber \\
 & \times \sum_{q\in\mathcal{Q}}\Big[\big\Vert\mathbf{f}_{q}\big\Vert^{2}+\sum\limits_{\varkappa\in\mathcal{M}_{\mathrm S}}\big(\varphi_{\mathrm{I},i,q,\varkappa}^{2}+\bar{\varphi}_{\mathrm{I},i,q,\varkappa}^{2}\big)\Big]+\sigma_{\mathrm{I},i}^{2},\label{eq:P1-comm-A}\\
 & \tau_{\mathrm{I},i,\jmath}\geq\zeta_{\mathrm{I},i,\jmath}\big(\mathbf{f}_{\jmath},\bm{\theta};\mathbf{f}_{\jmath}^{(\upsilon)},\bm{\theta}^{(\upsilon)}\big),\forall\jmath\in\mathcal{Q}\setminus \{i\}, \label{eq:P1-comm-B}\\
 &\tau_{\mathrm{I},i,\jmath}\geq\zeta_{\mathrm{I},i,\jmath}\big(-\mathbf{f}_{\jmath},\bm{\theta};-\mathbf{f}_{\jmath}^{(\upsilon)},\bm{\theta}^{(\upsilon)}\big),\forall\jmath\in\mathcal{Q}\setminus \{i\}, \label{eq:P1-comm-C}\\
 & \bar{\tau}_{\mathrm{I},i,\jmath}\geq\bar{\zeta}_{\mathrm{I},i,\jmath}\big( j\mathbf{f}_{\jmath},\bm{\theta}; j\mathbf{f}_{\jmath}^{(\upsilon)},\bm{\theta}^{(\upsilon)}\big),\forall\jmath\in\mathcal{Q}\setminus \{i\},\label{eq:P1-comm-D}\\
 & \bar{\tau}_{\mathrm{I},i,\jmath}\geq\bar{\zeta}_{\mathrm{I},i,\jmath}\big(- j\mathbf{f}_{\jmath},\bm{\theta};- j\mathbf{f}_{\jmath}^{(\upsilon)},\bm{\theta}^{(\upsilon)}\big),\forall\jmath\in\mathcal{Q}\setminus \{i\},\label{eq:P1-comm-E}\\
 & \varphi_{\mathrm{I},i,q,\varkappa}\geq\mu_{\mathrm{I},i,q,\varkappa}\big(\mathbf{f}_{q},\bm{\theta};-\mathbf{f}_{q}^{(\upsilon)},\bm{\theta}^{(\upsilon)}\big),
 \notag \\ & \qquad \qquad \qquad \qquad \qquad \qquad \qquad
 \forall q\in\mathcal{Q},\varkappa\in\mathcal{M}_{\mathrm S},\label{eq:P1-comm-F}\\
 & \varphi_{\mathrm{I},i,q,\varkappa}\geq\mu_{\mathrm{I},i,q,\varkappa}\big(-\mathbf{f}_{q},\bm{\theta};-\mathbf{f}_{q}^{(\upsilon)},\bm{\theta}^{(\upsilon)}\big),
 \notag \\ & \qquad \qquad \qquad \qquad \qquad \qquad \qquad
 \forall q\in\mathcal{Q},\varkappa\in\mathcal{M}_{\mathrm S},\label{eq:P1-comm-G}\\
 & \bar{\varphi}_{\mathrm{I},i,q,\varkappa}\geq\bar{\mu}_{\mathrm{I},i,q,\varkappa}\big( j\mathbf{f}_{q},\bm{\theta}; j\mathbf{f}_{q}^{(\upsilon)},\bm{\theta}^{(\upsilon)}\big),
  \notag \\ & \qquad \qquad \qquad \qquad \qquad \qquad \qquad
 \forall q\in\mathcal{Q},\varkappa\in\mathcal{M}_{\mathrm S}, \label{eq:P1-comm-H} \\
 & \bar{\varphi}_{\mathrm{I},i,q,\varkappa}\geq\bar{\mu}_{\mathrm{I},i,q,\varkappa}\big(- j\mathbf{f}_{q},\bm{\theta};- j\mathbf{f}_{q}^{(\upsilon)},\bm{\theta}^{(\upsilon)}\big),
\notag \\ & \qquad \qquad \qquad \qquad \qquad \qquad \qquad
\forall q\in\mathcal{Q},\varkappa\in\mathcal{M}_{\mathrm S}.\label{eq:P1-comm-I}
\end{align}
\end{subequations}
\end{theorem}
\begin{IEEEproof}
Please see Appendix~\ref{sec:proof-comm-Convex}.
\end{IEEEproof}
% Next, we note that the threshold on the harvested power at $e$-th ER is generally very small, and this can create numerical issues during the optimization process. To tackle this, we transform the constraints in~\eqref{eq:P1-EH} as follows:
Next, we observe that the harvested-power QoS thresholds at the $e$-th ER are typically several orders of magnitude smaller than the transmit power budget and interference terms. This disparity in magnitude can adversely affect the numerical conditioning of the resulting subproblems and lead to instability in off-the-shelf convex solvers. To improve numerical robustness while preserving equivalence, we reformulate the energy-harvesting constraints in~\eqref{eq:P1-EH} into a scaled, numerically better-conditioned form as follows:
\begin{align}
\frac{1}{\mathcal{P}_{e}}P_{\mathrm{E},e}(\mathbf{F},\bm{\theta})\geq1,\forall e\in\mathcal{E}.\label{eq:EH-normalized}
\end{align}
% Since the LHS of~\eqref{eq:EH-normalized} is non-concave, we obtain a corresponding concave lower bound using the following theorem. 
Since the LHS of~\eqref{eq:EH-normalized} is non-concave in the optimization variables $(\mathbf F, \bm \theta)$, we construct a locally tight concave minorant by invoking the following theorem, which yields a tractable inner approximation of the original constraint.
\begin{theorem}
\label{thm:EH-Convex}For a given $e\in\mathcal{E}$, a convex transformation of the constraints in~\eqref{eq:P1-EH} can be given by~\eqref{eq:P1-EH-CONVEX}, shown on the next page.  
%\begin{multline}
%\frac{\eta_{\mathrm{E},e}}{\mathcal{P}_{e}}\sum\limits_{q\in\mathcal{Q}}\Big[\psi_{\mathrm{E},e}\big(\mathbf{f}_{q},\bm{\theta};\mathbf{f}_{q}^{(\upsilon)},\bm{\theta}^{(\upsilon)}\big)+\varsigma^{2}\Big(2\Re\big\{\big(\mathbf{f}_{q}^{(\upsilon)}\big)^{\mathsf{H}}\mathbf{f}_{q}\big\}\\
%-\|\mathbf{f}_{q}^{(\upsilon)}\|^{2}+\sum\limits_{\varkappa\in\mathcal{M}_{\mathrm S}}\widetilde{\psi}_{\mathrm{E},e,\varkappa}\big(\mathbf{f}_{q},\bm{\theta};\mathbf{f}_{q}^{(\upsilon)},\bm{\theta}^{(\upsilon)}\big)\Big)\Big]\geq1.\label{eq:P1-EH-CONVEX}
%\end{multline}
\end{theorem}
\begin{IEEEproof}
Please see Appendix~\ref{sec:proof-EH-Convex}. 
\end{IEEEproof}
\begin{figure*}
	\begin{equation}
		\frac{\eta_{\mathrm{E},e}}{\mathcal{P}_{e}}\sum\nolimits_{q\in\mathcal{Q}}\Big[\psi_{\mathrm{E},e}\big(\mathbf{f}_{q},\bm{\theta};\mathbf{f}_{q}^{(\upsilon)},\bm{\theta}^{(\upsilon)}\big)+\varsigma^{2}\Big(2\Re\big\{\big(\mathbf{f}_{q}^{(\upsilon)}\big)^{\mathsf{H}}\mathbf{f}_{q}\big\}
		-\|\mathbf{f}_{q}^{(\upsilon)}\|^{2}+\sum\nolimits_{\varkappa\in\mathcal{M}_{\mathrm S}}\widetilde{\psi}_{\mathrm{E},e,\varkappa}\big(\mathbf{f}_{q},\bm{\theta};\mathbf{f}_{q}^{(\upsilon)},\bm{\theta}^{(\upsilon)}\big)\Big)\Big]\geq1.\label{eq:P1-EH-CONVEX}
	\end{equation}
\hrulefill 
\end{figure*}
% We now turn our attention on the information leakage constraints in~\eqref{eq:P1-leakage}. Note that although the LHS of~\eqref{eq:P1-leakage} is similar to that of~\eqref{eq:P1-comm}, the convexification process is different from that provided in~\textbf{Theorem~\ref{thm:comm-Convex}}, due to the reversed inequality in~\eqref{eq:P1-leakage} compared to that in~\eqref{eq:P1-comm}. The following theorem presents a convex equivalent of the constraints in~\eqref{eq:P1-leakage}. 

We now turn our attention to the information leakage constraints in~\eqref{eq:P1-leakage}. Although the LHS of~\eqref{eq:P1-leakage} has the same functional form as that of~\eqref{eq:P1-comm}, the convexification strategy cannot mirror that in~\textbf{Theorem~\ref{thm:comm-Convex}} because the inequality direction in~\eqref{eq:P1-leakage} is reversed relative to~\eqref{eq:P1-comm}, and thus defines a different type of level set. Consequently, a distinct reformulation is required. The following theorem provides a convex representation that is equivalent to the original information leakage constraints in~\eqref{eq:P1-leakage}.
\begin{theorem}
\label{thm:leakage-Convex}For a given $e\in\mathcal{E}$ and $i\in\mathcal{I}$,
a convex reformulation of the constraints in~\eqref{eq:P1-leakage}
is given by 
\begin{subequations}
\label{eq:P1-leakage-CONVEX}
\begin{align}
 & \sigma_{\mathrm{E},e}^{2}+\sum \nolimits_{\jmath\in\mathcal{Q} \setminus \{i\}}\psi_{\mathrm{E},e}\big(\mathbf{f}_{\jmath},\bm{\theta};\mathbf{f}_{\jmath}^{(\upsilon)},\bm{\theta}^{(\upsilon)}\big)
 \nonumber  \\ 
 & +\varsigma^{2}\sum \limits_{q\in\mathcal{Q}}\Big[2\Re\big\{\big(\mathbf{f}_{q}^{(\upsilon)}\big)^{\mathsf{H}}\mathbf{f}_{q}\big\} + \sum\limits_{\varkappa\in\mathcal{M}_{\mathrm S}}\widetilde{\psi}_{\mathrm{E},e,\varkappa}\big(\mathbf{f}_{q},\bm{\theta};\mathbf{f}_{q}^{(\upsilon)},\bm{\theta}^{(\upsilon)}\big)
 \nonumber \\ & 
\qquad \qquad \qquad -\|\mathbf{f}_{q}^{(\upsilon)}\|^{2}\Big] \geq \frac{1}{\Gamma_{\mathrm{E},e,i}}\big(\tau_{\mathrm{E},e,i}^{2}+\bar{\tau}_{\mathrm{E},e,i}^{2}\big),\label{eq:P1-leakage-a}\\
 % & \qquad\geq\frac{1}{\Gamma_{\mathrm{E},e,i}}\big(\tau_{\mathrm{E},e,i}^{2}+\bar{\tau}_{\mathrm{E},e,i}^{2}\big),\label{eq:P1-leakage-a}\\
& \tau_{\mathrm{E},e,i}\geq\ \zeta_{\mathrm{E},e,i}\big(\mathbf{f}_{i},\bm{\theta};\mathbf{f}_{i}^{(\upsilon)},\bm{\theta}^{(\upsilon)}\big),\label{eq:P1-leakage-b}\\
& \tau_{\mathrm{E},e,i}\geq\ \zeta_{\mathrm{E},e,i}\big(-\mathbf{f}_{i},\bm{\theta};-\mathbf{f}_{i}^{(\upsilon)},\bm{\theta}^{(\upsilon)}\big),\label{eq:P1-leakage-c}\\
& \bar{\tau}_{\mathrm{E},e,i}\geq\ \bar{\zeta}_{\mathrm{E},e,i}\big( j\mathbf{f}_{i},\bm{\theta};- j\mathbf{f}_{i}^{(\upsilon)},\bm{\theta}^{(\upsilon)}\big), \label{eq:P1-leakage-d} \\
& \bar{\tau}_{\mathrm{E},e,i}\geq\ \bar{\zeta}_{\mathrm{E},e,i}\big(- j\mathbf{f}_{i},\bm{\theta};- j\mathbf{f}_{i}^{(\upsilon)},\bm{\theta}^{(\upsilon)}\big).\label{eq:P1-leakage-e}
\end{align}
\end{subequations}
\end{theorem}
\begin{IEEEproof}
Please see Appendix~\ref{sec:proof-leakage-Convex}. 
\end{IEEEproof}
% We now shift our focus to the sensing constraints in~\eqref{eq:P1-sensing}. These constraints are relatively easy to handle because $\mathbf{c}_{l}$ is fixed. Note that both the numerator and denominator in~\eqref{eq:gamma_Bl} are convex w.r.t. $\mathbf{F}$. Therefore, for the $l$-th target, a convex equivalent of~\eqref{eq:P1-sensing} can be given by linearizing the numerator and some simple algebraic manipulations, which can be expressed using~\eqref{eq:gamma_Bl-Convex}, shown on the next page,  where $\varpi_{lq}^{(\upsilon)}\triangleq\mathbf{c}_{l}^{\mathsf{H}}\widehat{\mathbf{V}}_{l}\mathbf{f}_{q}^{(\upsilon)}$.
% yielding~\eqref{eq:gamma_Bl-Convex}, shown on the next page,

We now shift our focus to the sensing constraints in~\eqref{eq:P1-sensing}. These constraints are comparatively tractable because the receive combiners $\mathbf{c}_{l}$ are treated as fixed parameters. In this case, both the numerator and the denominator of the sensing SINR in~\eqref{eq:gamma_Bl} are convex functions of the transmit beamformers $\mathbf{F}$. Hence, for each target $l$, we construct a convex inner approximation of~\eqref{eq:P1-sensing} by linearizing the numerator around the current iterate $\mathbf{F}^{(\upsilon)}$ and performing suitable algebraic rearrangements of the resulting fractional inequality. This leads to the convex constraint in~\eqref{eq:gamma_Bl-Convex}, shown on the next page, where $\varpi_{lq}^{(\upsilon)}\triangleq\mathbf{c}_{l}^{\mathsf{H}}\widehat{\mathbf{V}}_{l}\mathbf{f}_{q}^{(\upsilon)}$.
\begin{figure*}
	\begin{multline} \label{eq:gamma_Bl-Convex}
		\frac{\bar{\alpha}}{\Gamma_{\mathrm{B},l}}\sum\nolimits_{q\in\mathcal{Q}}\Big[2\Re\big\{\big(\varpi_{lq}^{(\upsilon)}\big)^{\mathsf{H}}\mathbf{c}_{l}^{\mathsf{H}}\widehat{\mathbf{V}}_{l}\mathbf{f}_{q}\big\}-\big|\varpi_{lq}^{(\upsilon)}\big|^{2}\Big] \\ 
		\geq\sum\nolimits_{q\in\mathcal{Q}}\Big(\bar{\alpha}\varsigma^{2}\big\Vert\mathbf{c}_{l}\big\Vert^{2}\big\Vert\mathbf{f}_{q}\big\Vert^{2}+\big|\mathbf{c}_{l}^{\mathsf{H}}\mathbf{G}\mathbf{f}_{q}\big|^{2}\Big) +\bar{\alpha}\sum\nolimits_{\jmath\in\mathcal{L} \setminus \{l\}}\sum\nolimits_{q\in\mathcal{Q}}\big|\mathbf{c}_{l}^{\mathsf{H}}\widehat{\mathbf{V}}_{\jmath}\mathbf{f}_{q}\big|^{2}+\sigma_{\mathrm{B}}^{2}\big|\big|\mathbf{c}_{l}^{\mathsf{H}}\big|\big|^{2}.
	\end{multline}
\hrulefill 
\end{figure*}
To prevent the increasing problem dimension induced by the introduction of multiple slack variables, we handle~\eqref{eq:P1-UMC} via a penalty-based convex relaxation. In particular, we first relax the equality constraints in~\eqref{eq:P1-UMC} to one-sided convex inequalities, thereby obtaining a convex outer approximation of the original feasible set, and then augment the objective in~\eqref{eq:P2} with an appropriately weighted exact-penalty term that measures the residual of these inequalities. For a sufficiently large penalty parameter, any optimal solution of the penalized problem satisfies the relaxed constraints with equality, thus recovering feasibility w.r.t.~\eqref{eq:P1-UMC}.
% To further simplify the optimization objective, we introduce a slack variable $t$ to replace the term $\|\mathbf{F}\|_{\mathsf{F}}^{2}$.
As a consequence,~\eqref{eq:P2}  can be rewritten as
\begin{subequations}
\label{eq:P3}
\begin{align}
\underset{\mathbf{F},\bm{\theta},\mathcal{S}}{\min}\  & \|\mathbf{F}\|_{\mathsf{F}}^{2}-\wp \big[ 2 \Re \big\{ \big(\bm{\theta}^{(\upsilon)}\big)^{\mathsf H} \bm{\theta}\big\} - \| \bm{\theta}^{(\upsilon)} \|^2 \big] \label{eq:P3-objective}\\
\text{s.t.}\  
 & \eqref{eq:P1-comm-CONVEX},\forall i\in\mathcal{I},\label{eq:P3-comm}\\
 & \eqref{eq:P1-EH-CONVEX},\forall e\in\mathcal{E},\label{eq:P3-EH}\\
 & \eqref{eq:P1-leakage-CONVEX},\forall e\in\mathcal{E},i\in\mathcal{I},\label{eq:P3-leakage}\\
 & \eqref{eq:gamma_Bl-Convex},\forall l\in\mathcal{L},\label{eq:P3-sensing}\\
 & \sum\nolimits_{\mathsf{m}\in\{\mathsf{R},\mathsf{T}\}}|\theta_{\mathsf{m},\varkappa}|^{2}\leq1,\forall\varkappa\in\mathcal{M}_{\mathrm S},\label{eq:P3-relaxed-UMC}
\end{align}
\end{subequations}
where $\wp\geq0$ is the regularization/penalty parameter and $\mathcal S$ is the set of auxiliary variables. 
%$\mathcal{S} \triangleq \{\bm{\tau}_{\mathrm{I}} , \bar{\bm{\tau}}_{\mathrm{I}}, \mathbf{\Phi}_{\mathrm{I}}, \bar{\mathbf{\Phi}}_{\mathrm{I}},\bm{\tau}_{\mathrm{E}}, \bar{\bm{\tau}}_{\mathrm{E}}\}$ is the set of the auxiliary variables.
%with $\bm{\tau}_{\mathrm{I}} = \{\tau_{\mathrm{I},i,\jmath} \  | \  i\in\mathcal{I}, \jmath\in\mathcal{Q} \setminus \{i\} \}$,
%$\bar{\bm{\tau}}_{\mathrm{I}} = \{\bar{\tau}_{\mathrm{I},i,\jmath} \  | \  i\in\mathcal{I}, \jmath\in\mathcal{Q}\setminus \{i\} \}$,
%$\mathbf{\Phi}_{\mathrm{I}}= \{\varphi_{\mathrm{I},i,q,\varkappa} \  | \  i\in\mathcal{I}, q\in\mathcal{Q}, \varkappa\in\mathcal{M}_{\mathrm S} \}$,
%$ \bar{\mathbf{\Phi}}_{\mathrm{I}} = \{\bar{\varphi}_{\mathrm{I},i,q,\varkappa} \  | \ i\in\mathcal{I}, q\in\mathcal{Q}, \varkappa\in\mathcal{M}_{\mathrm S} \}$,
%$\bm{\tau}_{\mathrm{E}} = \{\tau_{\mathrm{E},e,i} \  | \ e\in\mathcal{E}, i\in\mathcal{I} \}$,
%and
%$\bar{\bm{\tau}}_{\mathrm{E}} = \{\bar{\tau}_{\mathrm{E},e,i} \  | \ e\in\mathcal{E}, i\in\mathcal{I}\}$.
% Note that~\eqref{eq:P3} is an SOCP problem that can be efficiently solved using off-the-shelf convex solvers, \emph{e.g.}, CVX, CVXPY or MOSEK. 
% Finally, the proposed AO-based algorithm for optimizing $\mathbf{F}$, $\bm{\theta}$, and $\mathbf{C}$ is summarized in \textbf{Algorithm~\ref{algo}}.
Note that~\eqref{eq:P3} is an SOCP and can be efficiently handled by interior-point methods implemented in standard convex optimization packages, \emph{e.g.}, CVX or CVXPY, together with state-of-the-art conic solvers such as MOSEK. The resulting AO-based procedure for jointly optimizing $\mathbf{F}$, $\bm{\theta}$, and $\mathbf{C}$ is summarized in \textbf{Algorithm~\ref{algo}}.

\begin{algorithm}[t]
\caption{The Proposed AO-based Algorithm to Solve~\eqref{eq:P1}}

\label{algo}

\KwIn{$\mathbf{F}^{(0)}$, $\bm{\theta}^{(0)}$}

$\upsilon\leftarrow 0$\;

\Repeat{convergence }{

Given $\mathbf{F}^{(\upsilon)}$ and $\bm{\theta}^{(\upsilon)}$, obtain $\mathbf{C}^{(\upsilon+1)}$ via~\eqref{eq:cl_opt}\;
% $\mathbf{c}_{l}^{(\upsilon)}, \forall l\in\mathcal{L}$
Given $\mathbf{C}^{(\upsilon+1)}$, solve~\eqref{eq:P3} to obtain $\mathbf{F}^{(\upsilon+1)}$,
$\bm{\theta}^{(\upsilon+1)}$\;

%Update: $\mathbf{F}^{(\upsilon+1)}\leftarrow\mathbf{F}_{\mathrm{opt}}$,
%$\bm{\theta}^{(\upsilon+1)}\leftarrow\bm{\theta}_{\mathrm{opt}}$\;

$\upsilon\leftarrow\upsilon+1$\;

}

$\mathbf{F}_{\mathrm{opt}}\leftarrow\mathbf{F}^{(\upsilon)}$, $\bm{\theta}_{\mathrm{opt}}\leftarrow\bm{\theta}^{(\upsilon)}$,
$\mathbf{C}_{\mathrm{opt}}\leftarrow\mathbf{C}^{(\upsilon)}$\;

\KwOut{$\mathbf{F}_{\mathrm{opt}}$, $\bm{\theta}_{\mathrm{opt}}$,
$\mathbf{C}_{\mathrm{opt}}$}
\end{algorithm}

\subsection{Convergence Analysis}
Note that the optimal receive beamforming matrix $\mathbf{C}$ is obtained via the \emph{generalized Rayleigh quotient maximization} method, and therefore the convergence behavior of the proposed AO-based algorithm in \textbf{Algorithm~\ref{algo}} is essentially determined by the convergence of subproblem~\eqref{eq:P3}. To facilitate analysis, we define the objective function of~\eqref{eq:P3} prior to linearizing the regularization term as $f(\mathbf{F},\bm{\theta})\triangleq \|\mathbf{F}\|_{\mathsf{F}}^{2}-\wp\| \bm{\theta} \|^2$. At the $\upsilon$-th iteration, let $f(\mathbf{F}^{(\upsilon)},\bm{\theta}^{(\upsilon)})$ denotes the optimal solution of~\eqref{eq:P3}. 
Since this solution remains feasible for the subsequent iteration, the sequence of objective values satisfies $f(\mathbf{F}^{(\upsilon+1)},\bm{\theta}^{(\upsilon+1)})  \leq f(\mathbf{F}^{(\upsilon)},\bm{\theta}^{(\upsilon)})$, which implies that the sequence $\{f(\mathbf{F}^{(\upsilon)},\bm{\theta}^{(\upsilon)})\}$ is monotonically non-increasing.
In addition, the non-negativity of $\|\mathbf{F}\|_{\mathsf{F}}^{2}$ together with the boundedness of $\| \bm{\theta} \|^2$ enforced by the power cone in~\eqref{eq:P3-relaxed-UMC} ensures that $f(\mathbf{F},\bm{\theta}) \geq - \wp M_\mathrm S$. Therefore, $\{f(\mathbf{F}^{(\upsilon)},\bm{\theta}^{(\upsilon)})\}$ is bounded from below and convergent by the monotone convergence theorem.
As for the convergence of the iterates, the AO framework guarantees that the sequence admits limit points, and any such limit point satisfies the first-order stationarity condition. Therefore, \textbf{Algorithm~\ref{algo}} is guaranteed to produce a convergent sequence of objective values, and all its limit points are stationary solutions of the original problem.

% \vspace{-1.1em}
\subsection{Complexity Analysis}
We now analyze the computational complexity of \textbf{Algorithm~\ref{algo}}. 
The complexity of each iteration is primarily dominated by the operations in lines 3 and 4 of the algorithm.
Specifically, the calculation of the
optimal $\mathbf{C}$ involves matrix multiplication, matrix inversion, and eigenvalue decomposition presented in~\eqref{eq:gamma_Bl_transform}, with a complexity of $\mathcal{O}(L(M_{\mathrm{R}}^3 + M_{\mathrm{T}}^2 M_{\mathrm{R}} + M_{\mathrm{T}} M_{\mathrm{R}}^2))$.
Subsequently, for line 4, the total number of real-valued optimization variables in~\eqref{eq:P3} is $2[M_{\mathrm{T}}Q+2M_{\mathrm S}+2I(Q-1)+2IQM_{\mathrm S}+2EI] + 1$ and the total number of conic constraints of second order is $I+4I(Q-1)+4IQM_{\mathrm S}+E+5EI+L+M_{\mathrm S}+1$. 
% Then, the size of each conic constraint is $4M_{\mathrm{T}}^2Q^2 + I[2(Q+M_{\mathrm{T}}Q+QM_{\mathrm S}-1)+QM_{\mathrm S}+1]^2+4I(Q-1)(2M_{\mathrm{T}}+1)^2+4IQM_{\mathrm S}(2M_{\mathrm{T}}+1)^2+E(M_{\mathrm{T}}+1)^2+EI(M_{\mathrm{T}}+3)^2+4EI(2M_{\mathrm{T}}+1)^2+L(M_{\mathrm{T}}Q+1)^2+16M_{\mathrm S}$.
Therefore, following~\cite[Sec.~6.6.2]{11-BenTal-Convex-Opt}, the per-iteration complexity for~\eqref{eq:P3} in line 4 is given by $\mathcal{O}((IQM_{\mathrm S}+EI)^{0.5} (IQM_{\mathrm S}+EI+M_{\mathrm{T}}Q) (I^2Q^2M_{\mathrm S}^2 + IQM_{\mathrm S} M_{\mathrm{T}}^2 + E^2I^2 + M_{\mathrm{T}}^2Q^2I + EIM_{\mathrm{T}}^2 + LM_{\mathrm{T}}^2Q^2 ) )$.
Since, in a practical setting, $M_{\mathrm S} \gg \max\{M_{\mathrm{T}}, M_{\mathrm{R}}, I, E, L, Q\}$, the overall complexity of each iteration for \textbf{Algorithm~\ref{algo}} can be approximated by $\mathcal{O}(M_{\mathrm S}^{3.5})$.
% $\mathcal{O}(2(4IQM_{\mathrm S}+5EI)^{0.5} (2IQM_{\mathrm S}+2EI+M_{\mathrm{T}}Q) [16(I^2Q^2M_{\mathrm S}^2 + IQM_{\mathrm S} M_{\mathrm{T}}^2 + E^2I^2) + 4M_{\mathrm{T}}^2Q^2I + 17 EIM_{\mathrm{T}}^2 + LM_{\mathrm{T}}^2Q^2 ] )$ 

\begin{figure}[t] 
	\centering
	\includegraphics[width=0.85\columnwidth, height=0.7\columnwidth]
	{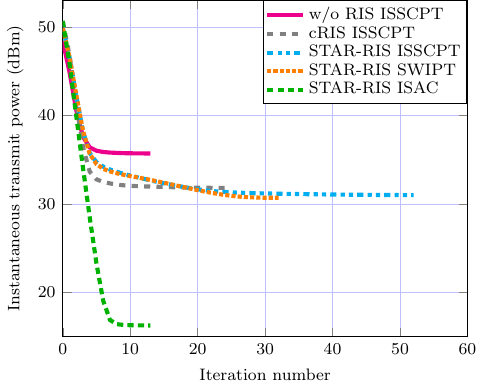}		
	\caption{Convergence behavior for the proposed AO-based algorithm.}
	\label{Fig_Convergence}
\end{figure}

\section{Simulation Results}
This section provides extensive numerical simulation results to evaluate the performance of the STAR-RIS-aided ISSCPT system.
Similar to~\cite{20-JSAC-SWIPT-MIMO-IRS-Pan}, we assume that B is located at the origin $(0,0)$~m, and S is deployed at $(5,2)$~m.
The IRs in the reflection and transmission regions are randomly distributed within a circle of radius 2~m centered at $(20,-1)$~m and $(20,5)$~m, respectively. Similarly, the ERs in both regions are randomly distributed within circles of radius 1 m centered at $(3,1)$~m and $(3,3)$~m, respectively. To ensure symmetry, the numbers of IRs and ERs are kept equal across the reflection and transmission regions.
The channel models $\mathbf{H}$, $\mathbf{h}_{\mathrm S,i},\forall i \in \mathcal{I}$, and $\mathbf{g}_{\mathrm S,e},\forall e \in \mathcal{E}$ follow a Rician distribution with a Rician factor of 3~dB and a path loss exponent of 2.2, while $\mathbf{h}_{\mathrm{B},i},\forall i \in \mathcal{I}$ and $\mathbf{g}_{\mathrm{B},e},\forall e \in \mathcal{E}$ are modeled as Rayleigh fading channels with a path loss exponent of 3.6.
Specifically, the large-scale path loss between two nodes with separation distance $d$~m is characterized by $(-30-10\beta \log_{10}(d/d_0))$~dB with $d_0 = 1$~m being the reference distance and $\beta$ denotes the path loss exponent.
The residual self-interference link 
$\mathbf{G}$ is modeled as a Rayleigh fading channel with variance equal to the noise variance.
In addition, two passive targets are considered at distances $15$~m and 
$18$~m from the BS, with azimuth angles of
$100^\circ$ and 
$150^\circ$, respectively.
Unless otherwise specified, we set 
$I = 2$, $E = 2$, $L = 2$, $M_{\mathrm{T}} = 10$, $M_{\mathrm{R}} = 4$, $M_{\mathrm S} = 64$, 
$\Gamma_{\mathrm{I},i} = \Gamma_{\mathrm{I}} = 10 \ \text{dB},\forall i\in\mathcal{I}$,
$\Gamma_{\mathrm{E},e,i} = \Gamma_{\mathrm{E}} = -10 \ \text{dB},\forall e\in\mathcal{E},i\in\mathcal{I}$,
$\Gamma_{\mathrm{B},l} = \Gamma_{\mathrm{B}} = 2 \ \text{dB},\forall l\in\mathcal{L}$,
$\mathcal{P}_{e} = \mathcal{P} = 10^{-4}~\mathrm{W},\forall e\in\mathcal{E}$,
$\bar{\alpha} = 0.5$,
$\eta_{\mathrm{E},e} = 0.8, \forall e\in\mathcal{E}$,
$\sigma_{\mathrm{I},i}^{2} =\sigma_{\mathrm{E},e}^{2} = \sigma_{\mathrm{B}}^{2} = \sigma^{2} = -90 \ \text{dBm},\forall i\in\mathcal{I},e\in\mathcal{E}$,
$\varsigma^{2} = 10 \sigma^{2}$,
and
$\wp = 0.01$.
Besides, the convergence tolerance is set to $10^{-3}$. 
The average transmit power in Figs.~\ref{Fig_mS}--\ref{Fig_sensingSINR} is obtained by averaging the instantaneous transmit power over 100 independent channel realizations.
For performance comparison of the proposed STAR-RIS-aided ISSCPT system, we consider the following benchmark systems: 

\begin{itemize}
	\item w/o RIS ISSCPT: This corresponds to a ISSCPT system without any metasurface. 
	\item cRIS ISSCPT: This corresponds to a cRIS-aided ISSCPT system. For a fair comparison with the STAR-RIS-aided counterpart, the cRIS is assumed to consist of $M_{\mathrm S}/2$ reflective and $M_{\mathrm S}/2$ refractive elements. Moreover, we consider perfect CSI for all the links in the case of cRIS ISSCPT. 
	\item STAR-RIS SWIPT: This corresponds to a STAR-RIS-aided secure SWIPT system, without any target. For such a system, the transmit signal from B is modeled as $\mathbf x = \sum_{i \in \mathcal I} \mathbf f_{\mathrm I, i} w_{\mathrm I, i} +  \sum_{e \in \mathcal E} \mathbf f_{\mathrm E, e} w_{\mathrm E, e}$. 
	\item STAR-RIS ISAC: This corresponds to a STAR-RIS-aided ISAC system, without any ER and information-leakage problem. In this case, the transmit signal from B is modeled as $\mathbf x = \sum_{i \in \mathcal I} \mathbf f_{\mathrm I, i} w_{\mathrm I, i} +  \sum_{l \in \mathcal L} \mathbf f_{\mathrm T, l} w_{\mathrm T, l}$. 
\end{itemize} 

Note that in Figs.~\ref{Fig_mR}–\ref{Fig_sensingSINR}, the average transmit power on the $y$-axis is abbreviated as “Avg. tx. power.” Moreover, in Figs.~\ref{Fig_mR}–\ref{Fig_commSINR} and Fig.~\ref{Fig_sensingSINR}, since the average transmit power required by the STAR-RIS-aided ISAC system is considerably smaller than that required by the other systems, a dual vertical-axis representation is used, where the performance of the STAR-RIS-aided ISAC system is indicated using the right-side vertical axis.

%For comparison, the performance of the proposed STAR-RIS-aided secure ISCPT system (STAR-RIS ISCPT) is benchmarked against the following schemes: a secure ISCPT system without RIS (w/o RIS ISCPT), a conventional RIS (cRIS)-assisted secure ISCPT system (cRIS ISCPT), a STAR-RIS-assisted secure SWIPT system (STAR-RIS SWIPT), and a STAR-RIS-assisted ISAC system (STAR-RIS ISAC).  
%In the case of the cRIS ISCPT scheme, we assume that the cRIS is partitioned into $M_{\mathrm S}/2$ reflective and $M_{\mathrm S}/2$ transmissive elements, yielding a total of $M_{\mathrm S}$. It is noteworthy that perfect CSI for all channels is assumed in this scheme.
%Unlike the preceding schemes, the STAR-RIS SWIPT scheme ignores the impact of targets and only imposes constraints on IRs and ERs.
%In contrast, the STAR-RIS ISAC scheme neglects ERs while focusing on the constraints associated with communication and sensing performance.

% \subsection{Algorithm Performance Evaluation}

\begin{figure}[t] 
	\centering
	\includegraphics[width=0.85\columnwidth, height=0.7\columnwidth]
	{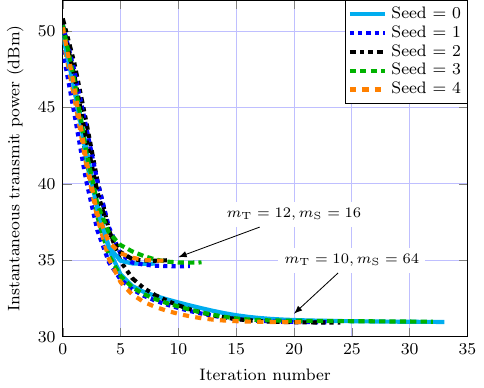}		
	\caption{Convergence behavior for different initial points (STAR-RIS ISSCPT).}
	\label{Fig_initialPoint}
\end{figure}
\emph{Convergence results:} In Fig.~\ref{Fig_Convergence}, we show the convergence behavior of the proposed AO-based algorithm for the STAR-RIS-aided ISSCPT system, compared with the baseline schemes. It can be noted that for the given set of channels, the proposed algorithm converges within a small number of iterations. For the ISSCPT system, the algorithm requires the most iterations to converge in the STAR-RIS-aided case, followed by the cRIS-aided case, while the no-RIS case converges with the fewest iterations. This is because the feasible constraint set is largest for STAR-RIS, smaller for cRIS, and smallest for the no-RIS system. A similar trend is observed for the number of required iterations in STAR-RIS-aided SWIPT and ISAC systems. It is also noteworthy that the proposed algorithm yields feasible points at each iteration, which allows for early termination if needed, for instance, under channels with very short coherence time.

% It can be clearly observed that among all ISCPT systems, the STAR-RIS scheme converges the slowest, requiring over 50 iterations. This is primarily because each element of the STAR-RIS necessitates simultaneous optimization of both transmission and reflection coefficients under the energy conservation constraint, which significantly increases the number of optimization variables and the degree of coupling. As a result, the step size adjustment becomes more sensitive, and more iterations are required to achieve stable convergence. Further comparison reveals that within the STAR-RIS schemes, the SWIPT system exhibits the second-largest number of iterations, only behind ISCPT. This is mainly attributed to the fact that the SWIPT system must guarantee sufficient energy harvesting for ERs while ensuring secure communication for IRs. Compared with the rate constraints of IRs, the constraints imposed on ERs are often more stringent and are highly coupled with the beamforming design, thereby substantially increasing the non-convexity and convergence difficulty of the problem. Therefore, the number of iterations for SWIPT is second only to ISCPT. In contrast, the ISAC system only needs to balance communication and sensing, with relatively simpler constraints and the lowest degree of coupling, thereby achieving the fastest convergence.

\emph{Sensitivity to initialization:} Since the proposed algorithm is based on the SCA-based framework, its performance can potentially be sensitive to the initial points. To examine this effect, Fig.~\ref{Fig_initialPoint} illustrates the impact of initial points for $\mathbf{F}$ and $\bm{\theta}$ on the convergence performance of the STAR-RIS ISSCPT system, where ``Seed'' denotes the random seed used to generate different $\big(\mathbf F^{(0)}, \bm{\theta}^{(0)}\big)$. As observed from the figure, the proposed algorithm achieves almost identical transmit power at convergence under different initialization, which demonstrates that the proposed algorithm is  insensitive to initialization. Furthermore, it can be noted that when $M_{\mathrm S} = 64$, the number of iterations required is significantly higher compared to the case with $M_{\mathrm S} = 16$. This is because an increase in the number of STAR-RIS elements introduces more optimization variables, thereby necessitating additional iterative steps for the convergence process.

% \subsection{Impact of System Parameters}

\begin{figure}[t] 
	\centering
	\includegraphics[width=0.85\columnwidth, height=0.7\columnwidth]
	{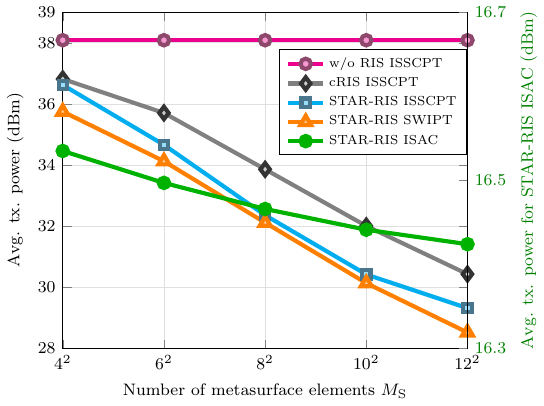}		
	\caption{Average transmit power for different values of $M_{\mathrm S}$.}
	\label{Fig_mS}
\end{figure}

\emph{Impact of the number of RIS elements $M_{\mathrm S}$:} In Fig.~\ref{Fig_mS}, we illustrate the impact of the number of RIS elements, $M_{\mathrm S}$, on the average transmit power requirement. For the no-RIS-aided ISSCPT system, the required transmit power remains constant due to the absence of a metasurface. In contrast, for all the other systems, the average transmit power decreases as the number of metasurface elements increases, highlighting the role of metasurfaces as enablers of green and secure integrated systems. This performance gain is attributed to the beamforming gain offered by the metasurface. Furthermore, as the number of metasurface elements grows, the performance gap between the no-RIS ISSCPT and STAR-RIS ISSCPT systems, as well as between the cRIS ISSCPT and STAR-RIS ISSCPT systems, becomes more pronounced, thereby establishing the clear advantage of employing STAR-RIS. For instance, when the number of metasurface elements increases from 64 to 100, the required transmit power for the STAR-RIS ISSCPT system is reduced by $\approx$15\% to $\approx$20\% compared to the no-RIS ISSCPT system, and by $\approx$4.4\% to $\approx$5\% compared to the cRIS ISSCPT system. However, integrating diverse services such as sensing, communication, and power transfer into a single system comes at the cost of additional transmit power. For instance, in the STAR-RIS-aided system with $M_{\mathrm S} = 64$, the ISSCPT scheme requires approximately 1\% more transmit power than the SWIPT system and about 85\% more transmit power than the ISAC system. These results clearly indicate that power transfer functionality is the dominant contributor to the total transmit power, and that metasurfaces are crucial to keep this requirement within practical limits in ISSCPT systems.

\emph{Impact of the number of transmit antennas $M_{\mathrm{T}}$:} Fig.~\ref{Fig_mT} depicts the impact of the number of transmit antennas at the DFBS, $M_{\mathrm{T}}$, on the average transmit power requirement. It is noticeable that as $M_{\mathrm{T}}$ increases, the required transmit power decreases. This reduction results from the higher spatial degrees of freedom and the active beamforming gain provided by more antennas, which enable more precise beamforming and improve transmission efficiency. It is also worth noting that increasing the number of transmit antennas provides greater benefits than increasing the number of metasurface elements, owing to the active beamforming gain of the former compared to the passive beamforming gain of the latter. For instance, from Figs.~\ref{Fig_mS} and~\ref{Fig_mT} for the STAR-RIS-aided ISSCPT system, increasing $M_{\mathrm S}$ from 16 to 36 (i.e., nearly doubling the number of metasurface elements), yields only about a $5.3\%$ reduction in the average transmit power requirement. In contrast, increasing $M_{\mathrm T}$ from 10 to 18 (i.e., less than doubling) achieves roughly a $7\%$ reduction. Nevertheless, deploying additional transmit antennas generally incurs a higher cost and hardware complexity compared to increasing the number of metasurface elements.
%However, a diminishing decrease is also observed as $M_{\mathrm{T}}$ increases further, since the beamforming gain gradually saturates, and each additional antenna contributes less to power savings.

\begin{figure}[t] 
	\centering
	\includegraphics[width=0.85\columnwidth, height=0.7\columnwidth]
	{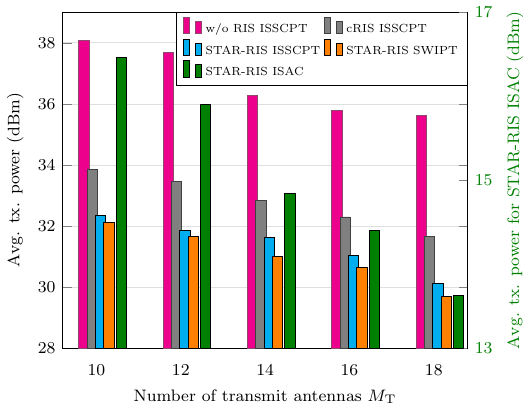}		
	\caption{Average transmit power for different values of $M_{\mathrm{T}}$.}
	\label{Fig_mT}
\end{figure}

\emph{Impact of the number of receive antennas $M_{\mathrm{R}}$:} In Fig.~\ref{Fig_mR}, we illustrate the impact of the number of receive antennas at the DFBS, $M_{\mathrm{R}}$, on the average transmit power. Since the variation in $M_{\mathrm{R}}$ only influences the sensing performance, and the SWIPT scheme does not support target sensing, its results are not included in this figure. It is worth noting that for ISSCPT systems, the gain from increasing $M_{\mathrm{R}}$ is relatively limited compared to that from increasing $M_{\mathrm{T}}$. This is because, under the considered system settings, the performance of ISSCPT systems is mainly constrained by the power harvesting requirements of the ERs, which are largely unaffected by the number of receive antennas. Consequently, the reduction in transmit power is less pronounced than with $M_{\mathrm{T}}$. In contrast, for the STAR-RIS-aided ISAC system, where the performance is dominated by sensing QoS constraints, the reduction in transmit power becomes much more significant. For example, when the number of receive antennas doubles from 4 to 8, the reduction in the required transmit power for the no-RIS ISSCPT, cRIS ISSCPT, STAR-RIS ISSCPT, and STAR-RIS ISAC systems is approximately $0.64\%$, $2.2\%$, $1.5\%$, and $21\%$, respectively.
%Similar treatment is applied in the subsequent Figs.~\ref{Fig_EH} and~\ref{Fig_sensingSINR}, where a certain scheme is not considered. The trend in Fig.~\ref{Fig_mR} is generally consistent with that in Fig.~\ref{Fig_mT} and will not be reiterated here. 
%This is because an increase in $M_{\mathrm{T}}$ simultaneously enhances communication, sensing, and power transfer performance, thereby significantly reducing the required transmit power. In contrast, increasing $M_{\mathrm{R}}$ only benefits sensing performance and does not directly enhance communication or power transfer. 

\begin{figure}[t] 
	\centering
	\includegraphics[width=0.85\columnwidth, height=0.7\columnwidth]	{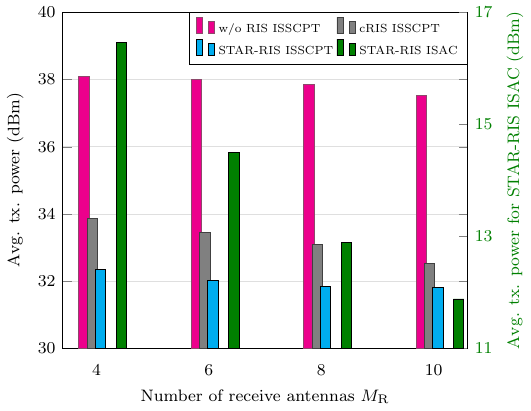}		
	\caption{Average transmit power for different values of $M_{\mathrm{R}}$.}
	\label{Fig_mR}
\end{figure}
\begin{figure}[t] 
	\centering
	\includegraphics[width=0.85\columnwidth, height=0.7\columnwidth]
	{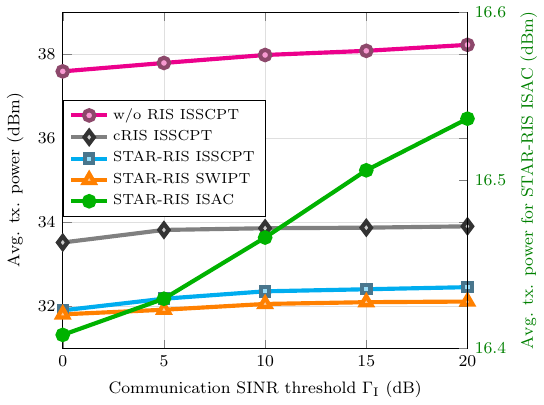}		
	\caption{Average transmit power gain for different values of $\Gamma_{\mathrm{I}}$.}
	\label{Fig_commSINR}
\end{figure}

\begin{figure}[t] 
	\centering
	\includegraphics[width=0.85\columnwidth, height=0.7\columnwidth]
	{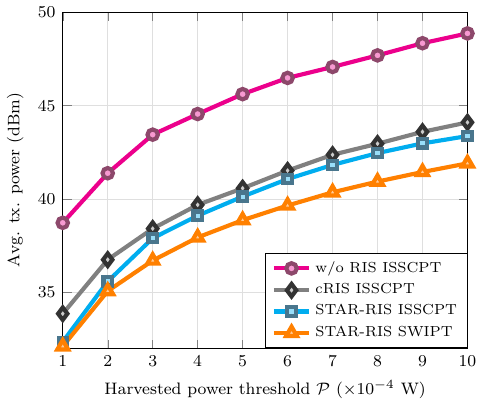}		
	\caption{Average transmit power for different values of $\mathcal{P}$.}
	\label{Fig_EH}
\end{figure}
\begin{figure}[t] 
	\centering
	\includegraphics[width=0.85\columnwidth, height=0.7\columnwidth]
	{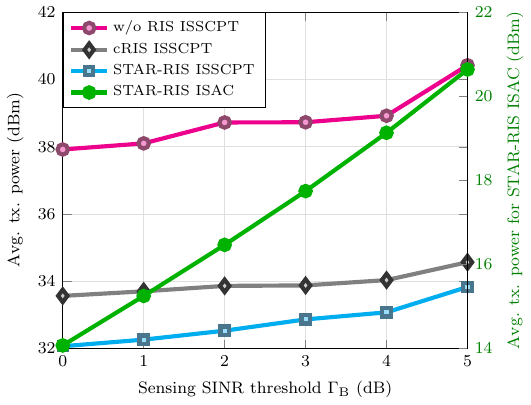}		
	\caption{Average transmit power for different values of $\Gamma_{\mathrm{B}}$.}
	\label{Fig_sensingSINR}
\end{figure}

\emph{Impact of the communication SINR threshold $\Gamma_{\mathrm{I}}$:} Fig.~\ref{Fig_commSINR} shows the impact of the communication SINR threshold for each IR, $\Gamma_{\mathrm{I}}$, on the average transmit power requirement. As $\Gamma_{\mathrm{I}}$ increases, the required transmit power of the system initially rises to satisfy the more stringent communication constraints. However, for ISSCPT systems, beyond a certain value of $\Gamma_{\mathrm{I}}$, the rate of increase in transmit power gradually decreases. This behavior is due to the joint effect of communication, energy harvesting, and sensing requirements: the beamforming design is already optimized to meet the sensing and power harvesting constraints while efficiently directing power toward the IRs, allowing further increases in $\Gamma_{\mathrm{I}}$ to be accommodated without a significant rise in total transmit power. Nevertheless, if $\Gamma_{\mathrm{I}}$ is increased well beyond this point, additional power would eventually be required, as the system reaches the limits of its beamforming and spatial degrees of freedom. This saturation effect highlights the efficiency of the STAR-RIS-assisted joint beamforming design and its ability to effectively reuse available power across multiple objectives. At the same time, it can also be observed that the STAR-RIS SWIPT system requires slightly lower transmit power compared to the STAR-RIS ISSCPT system due to the absence of sensing requirements, while the STAR-RIS ISAC system demands significantly lower transmit power. This highlights the dominance of the power harvesting constraint in determining the overall transmit power requirement of the ISSCPT and SWIPT systems under the considered settings.

%Fig.~\ref{Fig_commSINR} shows the impact of the communication SINR threshold for each IR, $\Gamma_{\mathrm{I}}$, on the average transmit power requirement. It is evident that as $\Gamma_{\mathrm{I}}$ increases, the required transmit power of the system also increases to meet the more stringent communication constraints. However, for the case of secure ISCPT systems, it is noteworthy that beyond a certain value of $\Gamma_{\mathrm{I}}$, the rate of increase in the transmit power requirement decreases. This is mainly because 
%satisfying the communication constraints in this system is relatively straightforward, so even as $\Gamma_{\mathrm{I}}$ continues to increase, the additional power required is relatively limited, resulting in a gradually flattening power growth curve.
%Moreover, it can be observed that the gap between the ISCPT and SWIPT schemes widens as $\Gamma_{\mathrm{I}}$ increases. This is due to the fact that in the ISCPT system, communication generates interference to sensing. As $\Gamma_{\mathrm{I}}$ rises, this interference becomes more pronounced, requiring extra transmit power to satisfy the sensing performance constraint. Consequently, compared with the SWIPT system, which only needs to satisfy communication and power transfer requirements, the ISCPT system exhibits significantly higher power demand at high thresholds.

\emph{Impact of the harvested power threshold $\mathcal{P}$:} In Fig.~\ref{Fig_EH}, we illustrate the impact of the harvested power threshold for each ER, $\mathcal{P}$, on the average transmit power. As expected, when $\mathcal{P}$ increases, the transmit power required by the system rises significantly. This observation validates the earlier conclusion that the power harvesting constraints of the ERs dominate the system performance, since additional transmit power must be consumed to satisfy the energy demands of all ERs as $\mathcal{P}$ increases.  Furthermore, this figure exhibits a trend similar to Fig.~\ref{Fig_commSINR}: as $\mathcal{P}$ increases, the power gap between the ISSCPT and SWIPT schemes gradually widens, and this effect is even more pronounced here. The reason is that, in the ISSCPT system, transmit power must be allocated to simultaneously satisfy both power harvesting and sensing requirements. Hence, under stringent power harvesting constraints, the ISSCPT scheme requires more transmit power than the SWIPT scheme, which only needs to satisfy the power harvesting requirements. Nevertheless, STAR-RIS consistently achieves lower transmit power than cRIS, even when cRIS is assumed to have perfect CSI. Moreover, a notably lower transmit power is required in the metasurface-assisted ISSCPT system compared to its no-RIS counterpart, further emphasizing the significance of incorporating metasurfaces into integrated system design.

\emph{Impact of the sensing SINR threshold $\Gamma_{\mathrm{B}}$:} Fig.~\ref{Fig_sensingSINR} displays the impact of the sensing SINR threshold, $\Gamma_{\mathrm{B}}$, on the average transmit power. As expected, an increase in $\Gamma_{\mathrm{B}}$ leads to a higher transmit power requirement for both the ISSCPT and ISAC systems. Interestingly, in the case of the ISSCPT system, the rate of increase in transmit power is relatively small for low-to-moderate values of $\Gamma_{\mathrm{B}}$, but grows significantly for larger values of $\Gamma_{\mathrm{B}}$. This occurs because, at low $\Gamma_{\mathrm{B}}$, the ISSCPT system operates in a power-harvesting-limited region, whereas for high $\Gamma_{\mathrm{B}}$ it transitions into a sensing-limited region. In contrast, the impact of increasing $\Gamma_{\mathrm{B}}$ is more pronounced in the STAR-RIS ISAC system over the range, reflecting the dominance of sensing constraints in this case.

%Similarly, as $\Gamma_{\mathrm{B}}$ increases, the required transmit power also rises. However, unlike the trends observed in Figs.~\ref{Fig_commSINR} and~\ref{Fig_EH}, the power growth here exhibits a steeper slope. This indicates that, since sensing performance is more sensitive to the quality of the received signal, stricter sensing constraints require greater power investment to effectively suppress interference and noise, thereby ensuring reliable target detection. Furthermore, similar to the phenomenon observed in Fig.~\ref{Fig_EH}, in the ISCPT system, increasing $\Gamma_{\mathrm{B}}$ gradually reduces the power gap between the STAR-RIS and cRIS schemes. The underlying reason is the same as in Fig.~\ref{Fig_EH} and is not repeated here. Nevertheless, the STAR-RIS scheme consistently maintains the optimal performance.

\section{Conclusion}
In this work, we studied the problem of transmit power minimization in a STAR-RIS-aided ISSCPT system under imperfect CSI, designing transmit, receive, and STAR-RIS beamforming jointly using AO, SCA, and SOCP. The results showed that ISSCPT generally requires more transmit power than ISAC and secure SWIPT, while STAR-RIS consistently outperforms both cRIS and no-RIS designs. We observed that transmit power grows with the communication SINR threshold but tends to saturate at high values, energy harvesting constraints dominate at low-to-moderate thresholds, and sensing SINR drives the steepest power increase at high thresholds. 
Overall, these findings highlight the benefits of STAR-RIS in reducing power consumption and provide useful insights for designing ISSCPT systems.
% Overall, these findings highlight the benefits of STAR-RIS in reducing power consumption and provide useful insights for designing integrated sensing, communication, and power transfer systems.
%In this paper, we have proposed a STAR-RIS-assisted secure ISCPT system that integrates secure communication, sensing, and wireless energy transfer into a unified framework.
%By jointly optimizing the transmit/receive beamforming and the STAR-RIS transmission/reflection coefficients via SCA, the system is capable of simultaneously guaranteeing communication secrecy for IRs, providing sufficient energy for ERs, and ensuring precise sensing performance.
%The corresponding optimization problem is solved via SCA combined with AO.
%The simulation results illustrated that STAR-RIS substantially outperforms conventional RIS and non-RIS schemes.

\appendices{}

\section{\label{sec:proof-comm-Convex}Proof of Theorem~\ref{thm:comm-Convex}}

For a given $i\in\mathcal{I}$, and using~\eqref{eq:gamma_Ii}, we
can rewrite~\eqref{eq:P1-comm} as 
\begin{align}
 & \frac{1}{\Gamma_{\mathrm{I},i}}\big|\widehat{\mathbf{z}}_{\mathrm{I},i}\mathbf{f}_{i}\big|^{2}\nonumber \\
\geq\  & \sum\limits_{\jmath\in\mathcal{Q} \setminus \{i\}}\big|\widehat{\mathbf{z}}_{\mathrm{I},i}\mathbf{f}_{\jmath}\big|^{2}+\varsigma^{2}\sum\limits_{q\in\mathcal{Q}}\Big(\big\Vert\mathbf{f}_{q}\big\Vert^{2}+\big\Vert\bm{\Theta}_{\mathrm{I},i}\mathbf{H}\mathbf{f}_{q}\big\Vert^{2}\Big)+\sigma_{\mathrm{I},i}^{2}\nonumber \\
=\  & \sum\nolimits_{\jmath\in\mathcal{Q} \setminus \{i\}}\big|\widehat{\mathbf{z}}_{\mathrm{I},i}\mathbf{f}_{\jmath}\big|^{2}+\sigma_{\mathrm{I},i}^{2}\nonumber \\
 & +\varsigma^{2}\sum\nolimits_{q\in\mathcal{Q}}\Big(\big\Vert\mathbf{f}_{q}\big\Vert^{2}+\sum\nolimits_{\varkappa\in\mathcal{M}_{\mathrm S}}|\theta_{\mathrm{I},i,\varkappa}\mathbf{h}_{\varkappa}\mathbf{f}_{q}|^{2}\Big),\label{eq:A-1}
\end{align}
where $\mathbf{h}_{\varkappa}\in\mathbb{C}^{1\times M_{\mathrm{T}}}$
is the $\varkappa$-th row of $\mathbf{H}$. Due to the bilinear coupling among the optimization variables, the functions on both sides of~\eqref{eq:A-1} are, in general, neither convex nor concave, so that the inequality defines a nonconvex constraint. To render this constraint amenable to the SCA framework, we first construct a locally tight concave lower bound (minorant) for the LHS of~\eqref{eq:A-1} as follows:
\begin{align}
 & \frac{1}{\Gamma_{\mathrm{I},i}}\big|\widehat{\mathbf{z}}_{\mathrm{I},i}\mathbf{f}_{i}\big|^{2}\geq\frac{1}{\Gamma_{\mathrm{I},i}}\Big[2\Re\big\{ \big(a_{\mathrm{I},i}^{(\upsilon)}\big)^{\mathsf{H}} \widehat{\mathbf{z}}_{\mathrm{I},i}\mathbf{f}_{i}\big\}-\big|a_{\mathrm{I},i}^{(\upsilon)}\big|^{2}\Big]\nonumber \\
=\  & \frac{1}{\Gamma_{\mathrm{I},i}}\Big[ 0.5\big\{\big\Vert a_{\mathrm{I},i}^{(\upsilon)}\widehat{\mathbf{z}}_{\mathrm{I},i}^{\mathsf{H}}+\mathbf{f}_{i}\big\Vert^{2}-\big\Vert a_{\mathrm{I},i}^{(\upsilon)}\widehat{\mathbf{z}}_{\mathrm{I},i}^{\mathsf{H}}-\mathbf{f}_{i}\big\Vert^{2}\big\}-\big|a_{\mathrm{I},i}^{(\upsilon)}\big|^{2}\Big] \nonumber \\
\geq\  & \frac{1}{\Gamma_{\mathrm{I},i}}\Big[\Re\Big\{\left(\mathbf{b}_{\mathrm{I},i}^{(\upsilon)}\right)^{\mathsf{H}}\big[a_{\mathrm{I},i}^{(\upsilon)}\widehat{\mathbf{z}}_{\mathrm{I},i}^{\mathsf{H}}+\mathbf{f}_{i}\big]\Big\}-0.5\big\Vert\mathbf{b}_{\mathrm{I},i}^{(\upsilon)}\big\Vert^{2}\nonumber \\
 & \qquad\qquad\qquad\qquad- 0.5\big\Vert a_{\mathrm{I},i}^{(\upsilon)}\widehat{\mathbf{z}}_{\mathrm{I},i}^{\mathsf{H}}-\mathbf{f}_{i}\big\Vert^{2}-\big|a_{\mathrm{I},i}^{(\upsilon)}\big|^{2}\Big]\nonumber \\
\triangleq\  & \frac{1}{\Gamma_{\mathrm{I},i}}\psi_{\mathrm{I},i}\big(\mathbf{f}_{i},\bm{\theta};\mathbf{f}_{i}^{(\upsilon)},\bm{\theta}^{(\upsilon)}\big),\label{eq:A-2}
\end{align}
where $a_{\mathrm{I},i}^{(\upsilon)}\triangleq\widehat{\mathbf{z}}_{\mathrm{I},i}^{(\upsilon)}\mathbf{f}_{i}^{(\upsilon)}$ and 
$\mathbf{b}_{\mathrm{I},i}^{(\upsilon)}\triangleq a_{\mathrm{I},i}^{(\upsilon)}\big(\widehat{\mathbf{z}}_{\mathrm{I},i}^{(\upsilon)}\big)^{\mathsf{H}}+\mathbf{f}_{i}^{(\upsilon)}$. 
The first inequality is obtained by replacing $\big|\widehat{\mathbf{z}}_{\mathrm{I},i}\mathbf{f}_{i}\big|^{2}$ with its first-order (affine) approximation around $a_{\mathrm{I},i}^{(\upsilon)}$, while the equality follows from the standard identity $\Re\big\{\mathbf{x}^{\mathsf{H}}\mathbf{y}\big\}= \frac{1}{4}\big(\|\mathbf{x}+\mathbf{y}\|^{2}-\|\mathbf{x}-\mathbf{y}\|^{2}\big)$. However, the RHS of this equality is still non-concave due to the presence of the convex term 
$\big\Vert a_{\mathrm{I},i}^{(\upsilon)}\widehat{\mathbf{z}}_{\mathrm{I},i}^{\mathsf{H}}+\mathbf{f}_{i}\big\Vert^{2}$. 
To obtain a concave minorant, we further linearize 
$\big\Vert a_{\mathrm{I},i}^{(\upsilon)}\widehat{\mathbf{z}}_{\mathrm{I},i}^{\mathsf{H}}+\mathbf{f}_{i}\big\Vert^{2}$ 
via its first-order Taylor expansion around $\mathbf{b}_{\mathrm{I},i}^{(\upsilon)}$ in the second inequality, which yields the concave lower bound $\psi_{\mathrm{I},i}\big(\mathbf{f}_{i},\bm{\theta};\mathbf{f}_{i}^{(\upsilon)},\bm{\theta}^{(\upsilon)}\big)$.
Next, we observe that the terms $\big|\widehat{\mathbf{z}}_{\mathrm{I},i}\mathbf{f}_{\jmath}\big|^{2}$
and $\big|\theta_{\mathrm{I},i,\varkappa}\mathbf{h}_{\varkappa}\mathbf{f}_{q}\big|^{2}$
on the RHS of~\eqref{eq:A-1} are nonconvex quadratic
functions, since they depend bilinearly on the optimization variables.
As a result, the RHS of~\eqref{eq:A-1} is neither convex
nor concave. To obtain a tractable conservative approximation compatible
with the SCA framework, we construct a convex upper bound for the
RHS of~\eqref{eq:A-1}. In particular, we replace~\eqref{eq:A-1}
with the following set of inequalities:
\begin{subequations}
\label{eq:A-3}
\begin{align}
 & \frac{1}{\Gamma_{\mathrm{I},i}}\psi_{\mathrm{I},i}\big(\mathbf{f}_{i},\bm{\theta};\mathbf{f}_{i}^{(\upsilon)},\bm{\theta}^{(\upsilon)}\big)\geq\sum\limits_{\jmath\in\mathcal{Q} \setminus \{i\}}\big(\tau_{\mathrm{I},\jmath}^{2}+\bar{\tau}_{\mathrm{I},\jmath}^{2}\big) +\varsigma^{2}\nonumber \\
 & \times \sum_{q\in\mathcal{Q}}\Big[\big\Vert\mathbf{f}_{q}\big\Vert^{2}+\sum\limits_{\varkappa\in\mathcal{M}_{\mathrm S}}\big(\varphi_{\mathrm{I},i,q,\varkappa}^{2}+\bar{\varphi}_{\mathrm{I},i,q,\varkappa}^{2}\big)\Big] +\sigma_{\mathrm{I},i}^{2},\label{eq:A-3-1}\\
 & \tau_{\mathrm{I},i,\jmath}\geq|\Re\{\widehat{\mathbf{z}}_{\mathrm{I},i}\mathbf{f}_{\jmath}\}|,\forall\jmath\in\mathcal{Q} \setminus \{i\},\label{eq:A-3-2}\\
 & \bar{\tau}_{\mathrm{I},i,\jmath}\geq|\Im\{\widehat{\mathbf{z}}_{\mathrm{I},i}\mathbf{f}_{\jmath}\}|,\forall\jmath\in\mathcal{Q}\setminus \{i\},\label{eq:A-3-3}\\
 & \varphi_{\mathrm{I},i,q,\varkappa}\geq|\Re\{\theta_{\mathrm{I},i,\varkappa}\mathbf{h}_{\varkappa}\mathbf{f}_{q}\}|,\forall q\in\mathcal{Q},\varkappa\in\mathcal{M}_{\mathrm S},\label{eq:A-3-4}\\
 & \bar{\varphi}_{\mathrm{I},i,q,\varkappa}\geq|\Im\{\theta_{\mathrm{I},i,\varkappa}\mathbf{h}_{\varkappa}\mathbf{f}_{q}\}|,\forall q\in\mathcal{Q},\varkappa\in\mathcal{M}_{\mathrm S},\label{eq:A-3-5}
\end{align}
\end{subequations}
where $\tau_{\mathrm{I},\jmath}$, $\bar{\tau}_{\mathrm{I},\jmath}$,
$\varphi_{\mathrm{I},i,q,\varkappa}$, and $\bar{\varphi}_{\mathrm{I},i,q,\varkappa}$
are auxiliary slack variables introduced to decouple the nonconvex terms.
It is straightforward to verify that the LHS of~\eqref{eq:A-3-1} is jointly
concave in $(\mathbf{F},\bm{\theta})$, whereas the RHS is convex in $\mathbf{F}$
and independent of $\bm{\theta}$. By contrast, for~\eqref{eq:A-3-2}–\eqref{eq:A-3-5}
we still need to construct convex upper surrogates for the corresponding
right-hand sides. Leveraging the standard epigraph reformulation
$a \geq |b| \Leftrightarrow (a \geq b)\land(a \geq -b)$, the constraint
in~\eqref{eq:A-3-2}, for a given $\jmath \in \mathcal{Q} \setminus \{i\}$,
can be equivalently expressed as
\begin{subequations}
\label{eq:A-4}
\begin{align}
\!\!\!\tau_{\mathrm{I},i,\jmath}\geq & \Re\{\widehat{\mathbf{z}}_{\mathrm{I},i}\mathbf{f}_{\jmath}\}\!=\! 0.25\big(\big\Vert\widehat{\mathbf{z}}_{\mathrm{I},i}^{\mathsf{H}}+\mathbf{f}_{\jmath}\big\Vert^{2}-\big\Vert\widehat{\mathbf{z}}_{\mathrm{I},i}^{\mathsf{H}}-\mathbf{f}_{\jmath}\big\Vert^{2}\big),\label{eq:A-4-1}\\
\!\!\!\tau_{\mathrm{I},i,\jmath}\geq & \!-\!\Re\{\widehat{\mathbf{z}}_{\mathrm{I},i}\mathbf{f}_{\jmath}\}\!=\! 0.25\big(\big\Vert\widehat{\mathbf{z}}_{\mathrm{I},i}^{\mathsf{H}}-\mathbf{f}_{\jmath}\big\Vert^{2}-\big\Vert\widehat{\mathbf{z}}_{\mathrm{I},i}^{\mathsf{H}}+\mathbf{f}_{\jmath}\big\Vert^{2}\big).\label{eq:A-4-2}
\end{align}
\end{subequations}
However, due to the presence of negative quadratic terms, the expressions on the RHS of~\eqref{eq:A-4} constitute a difference-of-convex (DC) function and are therefore nonconvex. To obtain a tractable convex surrogate compatible with the SCA framework, we linearize the concave quadratic components via their first-order Taylor expansions around the current iterate, which yields the following set of convex inequalities:
\begin{subequations}
\label{eq:A-5}
\begin{align}
\tau_{\mathrm{I},i,\jmath}&\geq 0.25 \Big[\big\Vert\widehat{\mathbf{z}}_{\mathrm{I},i}^{\mathsf{H}}+\mathbf{f}_{\jmath}\big\Vert^{2}+\big\Vert\big(\widehat{\mathbf{z}}_{\mathrm{I},i}^{(\upsilon)}\big)^{\mathsf{H}}-\mathbf{f}_{\jmath}^{(\upsilon)}\big\Vert^{2}
\nonumber \\ & \qquad\qquad\qquad 
-2\Re\big\{\big[\widehat{\mathbf{z}}_{\mathrm{I},i}^{(\upsilon)}-\big(\mathbf{f}_{\jmath}^{(\upsilon)}\big)^{\mathsf{H}}\big]\big(\widehat{\mathbf{z}}_{\mathrm{I},i}^{\mathsf{H}}-\mathbf{f}_{\jmath}\big)\big\}\Big]
\nonumber \\ & \triangleq\zeta_{\mathrm{I},i,\jmath}\big(\mathbf{f}_{\jmath},\bm{\theta};\mathbf{f}_{\jmath}^{(\upsilon)},\bm{\theta}^{(\upsilon)}\big),\label{eq:A-5-1}\\
\tau_{\mathrm{I},i,\jmath}& \geq 0.25 \Big[\big\Vert\widehat{\mathbf{z}}_{\mathrm{I},i}^{\mathsf{H}}-\mathbf{f}_{\jmath}\big\Vert^{2}+\big\Vert\big(\widehat{\mathbf{z}}_{\mathrm{I},i}^{(\upsilon)}\big)^{\mathsf{H}}+\mathbf{f}_{\jmath}^{(\upsilon)}\big\Vert^{2}
\nonumber \\ & \qquad\qquad\qquad 
-2\Re\big\{\big[\widehat{\mathbf{z}}_{\mathrm{I},i}^{(\upsilon)}+\big(\mathbf{f}_{\jmath}^{(\upsilon)}\big)^{\mathsf{H}}\big]\big(\widehat{\mathbf{z}}_{\mathrm{I},i}^{\mathsf{H}}+\mathbf{f}_{\jmath}\big)\big\}\Big]
\nonumber \\ &
\triangleq\zeta_{\mathrm{I},i,\jmath}\big(-\mathbf{f}_{\jmath},\bm{\theta};-\mathbf{f}_{\jmath}^{(\upsilon)},\bm{\theta}^{(\upsilon)}\big).\label{eq:A-5-2}
\end{align}
\end{subequations}
% Therefore, a convex approximation to the constraint in~\eqref{eq:A-3-2}
% is given by
% \begin{align}
% \tau_{\mathrm{I},i,\jmath}\geq\zeta_{\mathrm{I},i,\jmath}\big(\pm\mathbf{f}_{\jmath},\bm{\theta};\pm\mathbf{f}_{\jmath}^{(\upsilon)},\bm{\theta}^{(\upsilon)}\big),\forall\jmath\in\mathcal{Q}/i.\label{eq:A-6}
% \end{align}

Following a similar line of arguments, one can convexify~\eqref{eq:A-3-3}–\eqref{eq:A-3-5}, for fixed $j \in \mathcal Q \setminus \{i\}$, $q \in \mathcal Q$, and $\varkappa \in \mathcal M_{\mathrm S}$
as follows: 
\begin{subequations}
	\label{eq:A-7}
	\begin{align}
		\bar{\tau}_{\mathrm{I},i,\jmath} \geq & \ \bar{\zeta}_{\mathrm{I},i,\jmath} \big(  j\mathbf{f}_{\jmath},\bm{\theta};  j\mathbf{f}_{\jmath}^{(\upsilon)},\bm{\theta}^{(\upsilon)}\big), \\
		\bar{\tau}_{\mathrm{I},i,\jmath} \geq & \ \bar{\zeta}_{\mathrm{I},i,\jmath} \big( - j\mathbf{f}_{\jmath},\bm{\theta};  -j\mathbf{f}_{\jmath}^{(\upsilon)},\bm{\theta}^{(\upsilon)}\big), \\
		\varphi_{\mathrm I, i, q, \varkappa} \geq & \ \mu_{\mathrm{I},i,q,\varkappa} \big(\mathbf{f}_{q},\bm{\theta}; \mathbf{f}_{q}^{(\upsilon)}\!,\bm{\theta}^{(\upsilon)} \big), \\
		\varphi_{\mathrm I, i, q, \varkappa} \geq & \ \mu_{\mathrm{I},i,q,\varkappa} \big(-\mathbf{f}_{q},\bm{\theta};-\mathbf{f}_{q}^{(\upsilon)}\!,\bm{\theta}^{(\upsilon)} \big), \\
		\bar{\varphi}_{\mathrm I, i, q, \varkappa} \geq & \ \bar{\mu}_{\mathrm{I},i,q,\varkappa} \big(\mathbf{f}_{q},\bm{\theta}; \mathbf{f}_{q}^{(\upsilon)}\!,\bm{\theta}^{(\upsilon)} \big), \\
		\bar{\varphi}_{\mathrm I, i, q, \varkappa} \geq & \ \bar{\mu}_{\mathrm{I},i,q,\varkappa} \big(-\mathbf{f}_{q},\bm{\theta};-\mathbf{f}_{q}^{(\upsilon)}\!,\bm{\theta}^{(\upsilon)} \big),
	\end{align}
\end{subequations}
where $\bar{\zeta}_{\mathrm{I},i,\jmath} \big(  j\mathbf{f}_{\jmath},\bm{\theta};  j\mathbf{f}_{\jmath}^{(\upsilon)},\bm{\theta}^{(\upsilon)}\big) = 0.25\Big[\big\Vert\widehat{\mathbf{z}}_{\mathrm{I},i}^{\mathsf{H}}- j\mathbf{f}_{\jmath}\big\Vert^{2}+\big\Vert\big(\widehat{\mathbf{z}}_{\mathrm{I},i}^{(\upsilon)}\big)^{\mathsf{H}}+ j\mathbf{f}_{\jmath}^{(\upsilon)}\big\Vert^{2}
-2\Re\big\{\big[\widehat{\mathbf{z}}_{\mathrm{I},i}^{(\upsilon)}- j\big(\mathbf{f}_{\jmath}^{(\upsilon)}\big)^{\mathsf{H}}\big]\big(\widehat{\mathbf{z}}_{\mathrm{I},i}^{\mathsf{H}}+ j\mathbf{f}_{\jmath}\big)\big\}\Big]$, 
%=======
%=======
$\bar{\zeta}_{\mathrm{I},i,\jmath} \big( - j\mathbf{f}_{\jmath},\bm{\theta};  -j\mathbf{f}_{\jmath}^{(\upsilon)},\bm{\theta}^{(\upsilon)}\big) =  0.25 \Big[\big\Vert\widehat{\mathbf{z}}_{\mathrm{I},i}^{\mathsf{H}}+ j\mathbf{f}_{\jmath}\big\Vert^{2}+\big\Vert\big(\widehat{\mathbf{z}}_{\mathrm{I},i}^{(\upsilon)}\big)^{\mathsf{H}}- j\mathbf{f}_{\jmath}^{(\upsilon)}\big\Vert^{2}-2\Re\big\{\big[\widehat{\mathbf{z}}_{\mathrm{I},i}^{(\upsilon)}+ j\big(\mathbf{f}_{\jmath}^{(\upsilon)}\big)^{\mathsf{H}}\big]\big(\widehat{\mathbf{z}}_{\mathrm{I},i}^{\mathsf{H}}- j\mathbf{f}_{\jmath}\big)\big\}\Big]$, 
%=======
%=======
$\mu_{\mathrm{I},i,q,\varkappa} \big(\mathbf{f}_{q},\bm{\theta}; \mathbf{f}_{q}^{(\upsilon)}\!,\bm{\theta}^{(\upsilon)} \big) = 0.25\Big[\big\Vert(\theta_{\mathrm{I},i,\varkappa}\mathbf{h}_{\varkappa})^{\mathsf{H}}+\mathbf{f}_{q}\big\Vert^{2}+\big\Vert(\theta_{\mathrm{I},i,\varkappa}^{(\upsilon)}\mathbf{h}_{\varkappa})^{\mathsf{H}}-\mathbf{f}_{q}^{(\upsilon)}\big\Vert^{2}-2\Re\big\{\big[\theta_{\mathrm{I},i,\varkappa}^{(\upsilon)}\mathbf{h}_{\varkappa}-\big(\mathbf{f}_{q}^{(\upsilon)}\big)^{\mathsf{H}}\big]\big[(\theta_{\mathrm{I},i,\varkappa}\mathbf{h}_{\varkappa})^{\mathsf{H}}-\mathbf{f}_{q}\big]\big\}\Big]$,
%=======
%=======
$\mu_{\mathrm{I},i,q,\varkappa} \big(-\mathbf{f}_{q},\bm{\theta};-\mathbf{f}_{q}^{(\upsilon)}\!,\bm{\theta}^{(\upsilon)} \big) = 0.25\Big[\big\Vert(\theta_{\mathrm{I},i,\varkappa}\mathbf{h}_{\varkappa})^{\mathsf{H}}-\mathbf{f}_{q}\big\Vert^{2}+\big\Vert(\theta_{\mathrm{I},i,\varkappa}^{(\upsilon)}\mathbf{h}_{\varkappa})^{\mathsf{H}}+\mathbf{f}_{q}^{(\upsilon)}\big\Vert^{2}-2\Re\big\{\big[\theta_{\mathrm{I},i,\varkappa}^{(\upsilon)}\mathbf{h}_{\varkappa}+\big(\mathbf{f}_{q}^{(\upsilon)}\big)^{\mathsf{H}}\big]\big[(\theta_{\mathrm{I},i,\varkappa}\mathbf{h}_{\varkappa})^{\mathsf{H}}+\mathbf{f}_{q}\big]\big\}\Big]$,
%=======
%=======
$\bar{\mu}_{\mathrm{I},i,q,\varkappa} \big(\mathbf{f}_{q},\bm{\theta}; \mathbf{f}_{q}^{(\upsilon)}\!,\bm{\theta}^{(\upsilon)} \big) = 0.25\big[\big\Vert(\theta_{\mathrm{I},i,\varkappa}\mathbf{h}_{\varkappa})^{\mathsf{H}}- j\mathbf{f}_{q}\big\Vert^{2}+\big\Vert(\theta_{\mathrm{I},i,\varkappa}^{(\upsilon)}\mathbf{h}_{\varkappa})^{\mathsf{H}}+ j\mathbf{f}_{q}^{(\upsilon)}\big\Vert^{2} -2\Re\big\{\big[\theta_{\mathrm{I},i,\varkappa}^{(\upsilon)}\mathbf{h}_{\varkappa}\!-\!j\big(\mathbf{f}_{q}^{(\upsilon)}\big)^{\mathsf{H}}\big]\big[(\theta_{\mathrm{I},i,\varkappa}^{(\upsilon)}\mathbf{h}_{\varkappa})^{\mathsf{H}}+ j\mathbf{f}_{q}\big]\big\}\Big]$, and 
%=======
%=======
$\bar{\mu}_{\mathrm{I},i,q,\varkappa} \big(-\mathbf{f}_{q},\bm{\theta};-\mathbf{f}_{q}^{(\upsilon)}\!,\bm{\theta}^{(\upsilon)} \big) = 0.25\Big[\big\Vert(\theta_{\mathrm{I},i,\varkappa}\mathbf{h}_{\varkappa})^{\mathsf{H}}+ j\mathbf{f}_{q}\big\Vert^{2}+\big\Vert(\theta_{\mathrm{I},i,\varkappa}^{(\upsilon)}\mathbf{h}_{\varkappa})^{\mathsf{H}}- j\mathbf{f}_{q}^{(\upsilon)}\big\Vert^{2}
-2\Re\big\{\big[\theta_{\mathrm{I},i,\varkappa}^{(\upsilon)}\mathbf{h}_{\varkappa}+ j\big(\mathbf{f}_{q}^{(\upsilon)}\big)^{\mathsf{H}}\big]\big[(\theta_{\mathrm{I},i,\varkappa}^{(\upsilon)}\mathbf{h}_{\varkappa})^{\mathsf{H}} \! - \! j\mathbf{f}_{q}\big]\big\}\Big]$.
It can be verified that the right-hand sides of~\eqref{eq:A-5} and~\eqref{eq:A-7} are jointly convex in $(\mathbf{F},\bm{\theta})$. Hence, by collecting the constraints in~\eqref{eq:A-3-1}, \eqref{eq:A-5}, and~\eqref{eq:A-7}, we obtain an equivalent convex reformulation of the original constraints in~\eqref{eq:P1-comm}, as stated in~\eqref{eq:P1-comm-CONVEX}, which concludes the proof.

\section{\label{sec:proof-EH-Convex}Proof of Theorem~\ref{thm:EH-Convex}}

Using~\eqref{eq:harvested_power_Ee} and~\eqref{eq:EH-normalized},
for a given $e\in\mathcal{E}$, we write 
\begin{multline}
\frac{\eta_{\mathrm{E},e}}{\mathcal{P}_{e}}\sum\nolimits_{q\in\mathcal{Q}}\Big[|\widehat{\mathbf{z}}_{\mathrm{E},e}\mathbf{f}_{q}|^{2}+\varsigma^{2}\Big(\big\Vert\mathbf{f}_{q}\big\Vert^{2}\\
+\sum\nolimits_{\varkappa\in\mathcal{M}_{\mathrm S}}|\theta_{\mathrm{E},e,\varkappa}\mathbf{h}_{\varkappa}\mathbf{f}_{q}|^{2}\Big)\Big]\geq1.\label{eq:B-1}
\end{multline}
Following the same reasoning as in~\eqref{eq:A-2}, we construct a concave minorant of the term $\big|\widehat{\mathbf{z}}_{\mathrm{E},e}\mathbf{f}_{q}\big|^{2}$, which is jointly concave in $(\mathbf{F},\bm{\theta})$, and can be written as 
\begin{align}
|\widehat{\mathbf{z}}_{\mathrm{E},e}\mathbf{f}_{q}|^{2}&\geq
\Re\Big\{\big(\mathbf{b}_{\mathrm{E},e}^{(\upsilon)}\big)^{\mathsf{H}}\big[a_{\mathrm{E},e}^{(\upsilon)}\widehat{\mathbf{z}}_{\mathrm{E},e}^{\mathsf{H}}+\mathbf{f}_{q}\big]\Big\}- \frac{1}{2}\big\Vert\mathbf{b}_{\mathrm{E},e}^{(\upsilon)}\big\Vert^{2}
\nonumber \\ & \qquad \qquad \qquad \qquad
- \frac{1}{2}\big\Vert a_{\mathrm{E},e}^{(\upsilon)}\widehat{\mathbf{z}}_{\mathrm{E},e}^{\mathsf{H}}-\mathbf{f}_{q}\big\Vert^{2}-\big|a_{\mathrm{E},e}^{(\upsilon)}\big|^{2}
\nonumber \\ &
\triangleq
\psi_{\mathrm{E},e}\big(\mathbf{f}_{q},\bm{\theta};\mathbf{f}_{q}^{(\upsilon)},\bm{\theta}^{(\upsilon)}\big),\label{eq:B-2}
\end{align}
where 
% $\psi_{\mathrm{E},e}\big(\mathbf{f}_{q},\bm{\theta};\mathbf{f}_{q}^{(\upsilon)},\bm{\theta}^{(\upsilon)}\big)\triangleq\Re\big\{\big(\mathbf{b}_{\mathrm{E},e}^{(\upsilon)}\big)^{\mathsf{H}}\big[a_{\mathrm{E},e}^{(\upsilon)}\widehat{\mathbf{z}}_{\mathrm{E},e}^{\mathsf{H}}+\mathbf{f}_{q}\big]\big\}- \frac{1}{2}\big\Vert\mathbf{b}_{\mathrm{E},e}^{(\upsilon)}\big\Vert^{2}- \frac{1}{2}\big\Vert a_{\mathrm{E},e}^{(\upsilon)}\widehat{\mathbf{z}}_{\mathrm{E},e}^{\mathsf{H}}-\mathbf{f}_{q}\big\Vert^{2}-\big|a_{\mathrm{E},e}^{(\upsilon)}\big|^{2}$,
$a_{\mathrm{E},e}^{(\upsilon)}\triangleq\widehat{\mathbf{z}}_{\mathrm{E},e}^{(\upsilon)}\mathbf{f}_{q}^{(\upsilon)}$
and $\mathbf{b}_{\mathrm{E},e}^{(\upsilon)}\triangleq a_{\mathrm{E},e}^{(\upsilon)}\big(\widehat{\mathbf{z}}_{\mathrm{E},e}^{(\upsilon)}\big)^{\mathsf{H}}+\mathbf{f}_{q}^{(\upsilon)}$.
Next, a first-order Taylor approximation of $\big\Vert\mathbf{f}_{q}\big\Vert^{2}$
can be given as 
\begin{align}
\big\Vert\mathbf{f}_{q}\big\Vert^{2}\geq2\Re\big\{\big(\mathbf{f}_{q}^{(\upsilon)}\big)^{\mathsf{H}}\mathbf{f}_{q}\big\}-\|\mathbf{f}_{q}^{(\upsilon)}\|^{2}.\label{eq:B-3}
\end{align}
At the end, similar to~\eqref{eq:A-2}, a concave lower bound on
the term $|\theta_{\mathrm{E},e,\varkappa}\mathbf{h}_{\varkappa}\mathbf{f}_{q}|^{2}$
can be given by 
\begin{align}
& |\theta_{\mathrm{E},e,\varkappa}\mathbf{h}_{\varkappa}\mathbf{f}_{q}|^{2} \nonumber \\
\geq & \ 
\Re\Big\{\big(\widetilde{\mathbf{b}}_{\mathrm{E},e,q,\varkappa}^{(\upsilon)}\big)^{\mathsf{H}}\big[\widetilde{a}_{\mathrm{E},e,q,\varkappa}^{(\upsilon)}(\theta_{\mathrm{E},e,\varkappa}\mathbf{h}_{\varkappa})^{\mathsf{H}}+\mathbf{f}_{q}\big]\Big\} -\|\widetilde{a}_{\mathrm{E},e,q,\varkappa}^{(\upsilon)}\|^{2}
\nonumber \\ 
& \qquad
- 0.5\big\Vert\widetilde{a}_{\mathrm{E},e,q,\varkappa}^{(\upsilon)}(\theta_{\mathrm{E},e,\varkappa}\mathbf{h}_{\varkappa})^{\mathsf{H}}-\mathbf{f}_{q}\big\Vert^{2}
- 0.5\big\Vert\widetilde{\mathbf{b}}_{\mathrm{E},e,q,\varkappa}^{(\upsilon)}\big\Vert^{2}
\nonumber \\ &
\triangleq
\widetilde{\psi}_{\mathrm{E},e,\varkappa}\big(\mathbf{f}_{q},\bm{\theta};\mathbf{f}_{q}^{(\upsilon)},\bm{\theta}^{(\upsilon)}\big),\label{eq:B-4}
\end{align}
where 
% $\widetilde{\psi}_{\mathrm{E},e,\varkappa}\big(\mathbf{f}_{q},\bm{\theta};\mathbf{f}_{q}^{(\upsilon)},\bm{\theta}^{(\upsilon)}\big)\triangleq\Re\big\{\big(\widetilde{\mathbf{b}}_{\mathrm{E},e,q,\varkappa}^{(\upsilon)}\big)^{\mathsf{H}}\big[\widetilde{a}_{\mathrm{E}eq\varkappa}^{(\upsilon)}(\theta_{\mathrm{E},e,\varkappa}\mathbf{h}_{\varkappa})^{\mathsf{H}}+\mathbf{f}_{q}\big]\big\}- \frac{1}{2}\big\Vert\widetilde{\mathbf{b}}_{\mathrm{E},e,q,\varkappa}^{(\upsilon)}\big\Vert^{2}- \frac{1}{2}\big\Vert\widetilde{a}_{\mathrm{E}eq\varkappa}^{(\upsilon)}(\theta_{\mathrm{E},e,\varkappa}\mathbf{h}_{\varkappa})^{\mathsf{H}}-\mathbf{f}_{q}\big\Vert^{2}-\|\widetilde{a}_{\mathrm{E},e,q,\varkappa}^{(\upsilon)}\|^{2}$,
$\widetilde{a}_{\mathrm{E},e,q,\varkappa}^{(\upsilon)}\triangleq(\theta_{\mathrm{E},e,\varkappa}^{(\upsilon)}\mathbf{h}_{\varkappa})\mathbf{f}_{q}^{(\upsilon)}$
and $\widetilde{\mathbf{b}}_{\mathrm{E},e,q,\varkappa}^{(\upsilon)}\triangleq\widetilde{a}_{\mathrm{E},e,q,\varkappa}^{(\upsilon)}(\theta_{\mathrm{E},e,\varkappa}^{(\upsilon)}\mathbf{h}_{\varkappa})^{\mathsf{H}}+\mathbf{f}_{q}^{(\upsilon)}$.
By aggregating~\eqref{eq:B-1}–\eqref{eq:B-4}, we obtain the convex representation of the energy-harvesting constraints in~\eqref{eq:P1-EH}, as given in~\eqref{eq:P1-EH-CONVEX}, which completes the proof. 

\section{\label{sec:proof-leakage-Convex}Proof of Theorem~\ref{thm:leakage-Convex}}

Using~\eqref{eq:gamma_Ee} and~\eqref{eq:P1-leakage}, for a given
$e\in\mathcal{E}$ and $i\in\mathcal{I}$, one can write
\begin{multline}
\sigma_{\mathrm{E},e}^{2}+\sum\nolimits_{\jmath\in\mathcal{Q} \setminus \{i\}}\big|\widehat{\mathbf{z}}_{\mathrm{E},e}\mathbf{f}_{\jmath}\big|^{2}\\
+\varsigma^{2}\sum\limits_{q\in\mathcal{Q}}\Big(\big\Vert\mathbf{f}_{q}\big\Vert^{2}+\big\Vert\bm{\Theta}_{\mathrm{E},e}\mathbf{H}\mathbf{f}_{q}\big\Vert^{2}\Big)\geq\frac{1}{\Gamma_{\mathrm{E},e,i}}|\widehat{\mathbf{z}}_{\mathrm{E},e}\mathbf{f}_{i}|^{2}.\label{eq:C-1}
\end{multline}
Since the preceding inequality defines a nonconvex constraint (neither side is convex or concave in the optimization variables), we construct a concave minorant for the left-hand side and a convex majorant for the right-hand side. In particular, by following the surrogate construction in Appendices~\ref{sec:proof-comm-Convex} and~\ref{sec:proof-EH-Convex}, we obtain the following tractable approximations: 
\begin{subequations}
\label{eq:C-3}
\begin{align}
 & \sigma_{\mathrm{E},e}^{2}+\sum \nolimits_{\jmath\in\mathcal{Q} \setminus \{i\}}\psi_{\mathrm{E},e}\big(\mathbf{f}_{\jmath},\bm{\theta};\mathbf{f}_{\jmath}^{(\upsilon)},\bm{\theta}^{(\upsilon)}\big)
\nonumber  \\ 
& +\varsigma^{2}\sum \limits_{q\in\mathcal{Q}}\Big[2\Re\big\{\big(\mathbf{f}_{q}^{(\upsilon)}\big)^{\mathsf{H}}\mathbf{f}_{q}\big\} + \sum\limits_{\varkappa\in\mathcal{M}_{\mathrm S}}\widetilde{\psi}_{\mathrm{E},e,\varkappa}\big(\mathbf{f}_{q},\bm{\theta};\mathbf{f}_{q}^{(\upsilon)},\bm{\theta}^{(\upsilon)}\big)
\nonumber \\ & 
\qquad \qquad \qquad  -\|\mathbf{f}_{q}^{(\upsilon)}\|^{2}\Big] \geq \frac{1}{\Gamma_{\mathrm{E},e,i}}\big(\tau_{\mathrm{E},e,i}^{2}+\bar{\tau}_{\mathrm{E},e,i}^{2}\big),\label{eq:C-3a}\\
 & \tau_{\mathrm{E},e,i}\geq|\Re\{\widehat{\mathbf{z}}_{\mathrm{E},e}\mathbf{f}_{i}\}|,\label{eq:C-3b}\\
 & \bar{\tau}_{\mathrm{E},e,i}\geq|\Im\{\widehat{\mathbf{z}}_{\mathrm{E},e}\mathbf{f}_{i}\}|,\label{eq:C-3c}
\end{align}
\end{subequations}
where $\psi_{\mathrm{E},e}(\cdot,\cdot;\cdot,\cdot)$ and $\widetilde{\psi}_{\mathrm{E},e}(\cdot,\cdot;\cdot,\cdot)$ are defined in Appendix~\ref{sec:proof-EH-Convex}, and $\tau_{\mathrm{E},e,i}$ and $\bar{\tau}_{\mathrm{E},e,i}$ are auxiliary slack variables. By construction, the LHS of~\eqref{eq:C-3a} is jointly concave in $(\mathbf{F},\bm{\theta})$, and the RHS is convex, so~\eqref{eq:C-3a} already defines a convex constraint. In contrast, the right-hand sides of~\eqref{eq:C-3b} and~\eqref{eq:C-3c} remain nonconvex due to the bilinear coupling between $\mathbf{F}$ and $\bm{\theta}$. To embed these constraints into the SCA framework, we construct locally tight convex majorants for the nonconvex terms, following the same surrogate-design principles as in Appendix~\ref{sec:proof-comm-Convex}, which yields
\begin{subequations}
	\label{eq:C-4}
	\begin{align}
		\tau_{\mathrm{E},e,i} & \geq \zeta_{\mathrm{E},e,i}\big(\mathbf{f}_{i},\bm{\theta};\mathbf{f}_{i}^{(\upsilon)},\bm{\theta}^{(\upsilon)}\big), \\
		\tau_{\mathrm{E},e,i} & \geq \zeta_{\mathrm{E},e,i}\big(-\mathbf{f}_{i},\bm{\theta};-\mathbf{f}_{i}^{(\upsilon)},\bm{\theta}^{(\upsilon)}\big), \\
		\bar{\tau}_{\mathrm{E},e,i} & \geq \bar{\zeta}_{\mathrm{E},e,i}\big( j\mathbf{f}_{i},\bm{\theta}; j\mathbf{f}_{i}^{(\upsilon)},\bm{\theta}^{(\upsilon)}\big), \\
		\bar{\tau}_{\mathrm{E},e,i} & \geq \bar{\zeta}_{\mathrm{E},e,i}\big( -j\mathbf{f}_{i},\bm{\theta}; -j\mathbf{f}_{i}^{(\upsilon)},\bm{\theta}^{(\upsilon)}\big),
	\end{align}
\end{subequations}
where $\zeta_{\mathrm{E},e,i}\big(\mathbf{f}_{i},\bm{\theta};\mathbf{f}_{i}^{(\upsilon)},\bm{\theta}^{(\upsilon)}\big) = 0.25\Big[\big\Vert\widehat{\mathbf{z}}_{\mathrm{E},e}^{\mathsf{H}}+\mathbf{f}_{i}\big\Vert^{2}+\big\Vert\big(\widehat{\mathbf{z}}_{\mathrm{E},e}^{(\upsilon)}\big)^{\mathsf{H}}+\mathbf{f}_{i}^{(\upsilon)}\big\Vert^{2}
-2\Re\big\{\big(\widehat{\mathbf{z}}_{\mathrm{E},e}^{(\upsilon)}-\big(\mathbf{f}_{i}^{(\upsilon)}\big)^{\mathsf{H}}\big)\big(\widehat{\mathbf{z}}_{\mathrm{E},e}^{\mathsf{H}}-\mathbf{f}_{i}\big)\big\}\Big]$, 
%========
%========
$\zeta_{\mathrm{E},e,i}\big(-\mathbf{f}_{i},\bm{\theta};-\mathbf{f}_{i}^{(\upsilon)},\bm{\theta}^{(\upsilon)}\big) = 0.25\Big[\big\Vert\widehat{\mathbf{z}}_{\mathrm{E},e}^{\mathsf{H}}-\mathbf{f}_{i}\big\Vert^{2}+\big\Vert\big(\widehat{\mathbf{z}}_{\mathrm{E},e}^{(\upsilon)}\big)^{\mathsf{H}}-\mathbf{f}_{i}^{(\upsilon)}\big\Vert^{2}
-2\Re\big\{\big(\widehat{\mathbf{z}}_{\mathrm{E},e}^{(\upsilon)}+\big(\mathbf{f}_{i}^{(\upsilon)}\big)^{\mathsf{H}}\big)\big(\widehat{\mathbf{z}}_{\mathrm{E},e}^{\mathsf{H}}+\mathbf{f}_{i}\big)\big\}\Big]$, 
%========
%========
$\bar{\zeta}_{\mathrm{E},e,i}\big( j\mathbf{f}_{i},\bm{\theta}; j\mathbf{f}_{i}^{(\upsilon)},\bm{\theta}^{(\upsilon)}\big) = 0.25\Big[\big\Vert\widehat{\mathbf{z}}_{\mathrm{E},e}^{\mathsf{H}}- j\mathbf{f}_{i}\big\Vert^{2}+\big\Vert\big(\widehat{\mathbf{z}}_{\mathrm{E},e}^{(\upsilon)}\big)^{\mathsf{H}}+ j\mathbf{f}_{i}^{(\upsilon)}\big\Vert^{2}
-2\Re\big\{\big(\widehat{\mathbf{z}}_{\mathrm{E},e}^{(\upsilon)}-\big( j\mathbf{f}_{i}^{(\upsilon)}\big)^{\mathsf{H}}\big)\big(\widehat{\mathbf{z}}_{\mathrm{E},e}^{\mathsf{H}}+ j\mathbf{f}_{i}\big)\big\}\Big]$, and 
%========
%========
$\bar{\zeta}_{\mathrm{E},e,i}\big( -j\mathbf{f}_{i},\bm{\theta}; -j\mathbf{f}_{i}^{(\upsilon)},\bm{\theta}^{(\upsilon)}\big) = 0.25\Big[\big\Vert\widehat{\mathbf{z}}_{\mathrm{E},e}^{\mathsf{H}}+ j\mathbf{f}_{i}\big\Vert^{2}+\big\Vert\big(\widehat{\mathbf{z}}_{\mathrm{E},e}^{(\upsilon)}\big)^{\mathsf{H}}- j\mathbf{f}_{i}^{(\upsilon)}\big\Vert^{2}
-2\Re\big\{\big(\widehat{\mathbf{z}}_{\mathrm{E},e}^{(\upsilon)}+\big( j\mathbf{f}_{i}^{(\upsilon)}\big)^{\mathsf{H}}\big)\big(\widehat{\mathbf{z}}_{\mathrm{E},e}^{\mathsf{H}}- j\mathbf{f}_{i}\big)\big\}\Big]$.
Using~\eqref{eq:C-3a} and \eqref{eq:C-4}, convex reformulation of the constraints in~\eqref{eq:P1-leakage} is given by~\eqref{eq:P1-leakage-CONVEX}. The proof is completed.

\balance

\bibliographystyle{IEEEtran}
\bibliography{references}

\end{document}